\documentclass[a4paper,UKenglish,cleveref, autoref, thm-restate]{lipics-v2019}

\usepackage{graphicx}
\usepackage{subcaption}
\usepackage{caption}
\usepackage{amstext,amssymb,amsthm,amsfonts,latexsym,amsmath}
\usepackage{mathtools}
\usepackage{algorithm}
\usepackage[noend]{algpseudocode}
\algnewcommand\algorithmicforeach{\textbf{for each}}
\algdef{S}[FOR]{ForEach}[1]{\algorithmicforeach\ #1\ \algorithmicdo}

\newtheorem{manualtheoreminner}{Theorem}
\newenvironment{customthm}[1]{%
  \renewcommand\themanualtheoreminner{#1}%
  \manualtheoreminner
}{\endmanualtheoreminner}

\newtheorem{manualcorollaryinner}{Corollary}
\newenvironment{customcorollary}[1]{%
  \renewcommand\themanualcorollaryinner{#1}%
  \manualcorollaryinner
}{\endmanualcorollaryinner}

\newtheorem{manuallemmainner}{Lemma}
\newenvironment{customlemma}[1]{%
  \renewcommand\themanuallemmainner{#1}%
  \manuallemmainner
}{\endmanuallemmainner}

\DeclarePairedDelimiter\abs{\lvert}{\rvert}%
\DeclarePairedDelimiter\norm{\lVert}{\rVert}%

\makeatletter
\let\oldabs\abs
\def\abs{\@ifstar{\oldabs}{\oldabs*}}
\let\oldnorm\norm
\def\norm{\@ifstar{\oldnorm}{\oldnorm*}}
\makeatother

\makeatletter
\DeclareRobustCommand{\iscircle}{\mathord{\mathpalette\is@circle\relax}}
\newcommand\is@circle[2]{%
  \begingroup
  \sbox\z@{\raisebox{\depth}{$\m@th#1\bigcirc$}}%
  \sbox\tw@{$#1\square$}%
  \resizebox{!}{\ht\tw@}{\usebox{\z@}}%
  \endgroup
}
\makeatother

\newtheorem{observation}[theorem]{Observation}

\title{Approximating the $\lambda$-low-density value} 

\titlerunning{Approximating the $\lambda$-low-density value} 

\author{Joachim Gudmundsson}{The University of Sydney}{joachim.gudmundsson@sydney.edu.au}{https://orcid.org/0000-0002-6778-7990
}{Funded by the Australian Government through the Australian Research Council DP180102870.}

\author{Zijin Huang}{The University of Sydney}{zijin.huang@sydney.edu.au}{https://orcid.org/0000-0003-3417-5303}{}

\author{Sampson Wong}{The University of Sydney}{swon7907@sydney.edu.au}{}{}

\authorrunning{J. Gudmundsson, Z. Huang and S. Wong} 

\Copyright{Joachim Gudmundsson, Zijin Huang and Sampson Wong} 

\ccsdesc[100]{Theory of computation~Design and analysis of algorithms} 

\keywords{realistic input models, packedness, low density, computational geometry
} 

\nolinenumbers 

\EventEditors{John Q. Open and Joan R. Access}
\EventNoEds{2}
\EventLongTitle{42nd Conference on Very Important Topics (CVIT 2016)}
\EventShortTitle{CVIT 2016}
\EventAcronym{CVIT}
\EventYear{2016}
\EventDate{December 24--27, 2016}
\EventLocation{Little Whinging, United Kingdom}
\EventLogo{}
\SeriesVolume{42}
\ArticleNo{23}

\begin{document}
\maketitle
\begin{abstract}
The use of realistic input models has gained popularity in the theory community. Assuming a realistic input model often precludes complicated hypothetical inputs, and the analysis yields bounds that better reflect the behaviour of algorithms in practice. 

One of the most popular models for polygonal curves and line segments is $\lambda$-low-density. To select the most efficient algorithm for a certain input, one often needs to compute the $\lambda$-low-density value, or at least an approximate value. In this paper, we show that given a set of $n$ line segments in $\mathbb{R}^2$ one can compute a $3$-approximation of the $\lambda$-low density value in $O(n \log n + \lambda n)$ time. We also show how to maintain a $3$-approximation of the $\lambda$-low density value while allowing insertions of new segments in $O(\log n + \lambda^2)$ amortized time per update. 

Finally, we argue that many real-world data sets have a small $\lambda$-low density value, warranting the recent development of specialised algorithms. This is done by computing approximate $\lambda$-low density values for $12$ real-world data sets.  
\end{abstract}

\section{Introduction}
\label{sec:Introduction}
When designing algorithms, researchers often use worst-case analysis to show the theoretical upper bound of their algorithms. The drawback of this approach is that these worst-case cases are often convoluted, and they are unlikely to occur in practice. Realistic input models attempt to rectify this issue by placing realistic constraints on the input, resulting in an analysis that better reflects the real-world performance.

There have been many proposed realistic models for geometric data including fatness,~low density, uncluttered, simple cover, and packedness~\cite{deBerg2002,driemel2012,stappen1994}, to name a few. These models place realistic constraints on the input. We give three examples of difficult computational tasks that become much more tractable under realistic input constraints.

In the case of calculating the Fréchet distance, Bringmann~\cite{Bringmann2014} showed a lower bound of $\Omega(n^{2 - \epsilon})$ for computing a $(2 - \epsilon)$-approximation of the Fréchet distance between two polygonal curves assuming the Strong Exponential Time Hypothesis (SETH). SETH asserts that $k$-SAT cannot be solved in $O((2 - \epsilon)^n)$ time for any constant $\epsilon > 0$, and $k$-SAT is an NP-complete problem. SETH is a popular conjecture, and it implies that there is no strongly subquadratic time algorithm to calculate the Fréchet distance. However, Driemel et al.~\cite{driemel2012} showed a near-linear time $(1+\epsilon)$-approximation algorithm in the case when the (polygonal) curves are $c$-packed or $\lambda$-low density. A curve $\pi$ is $c$-packed if the total length of $\pi$ inside any ball of radius $r$ is at most $c \cdot r$, and a set of objects is $\lambda$-low-density if for any ball of any size, the number of objects whose size is greater than the radius of the ball that intersect the ball is at most $\lambda$.    

In another instance, Van der Stappen~\cite{stappen1994} introduced $\lambda$-low-density as a realistic assumption for obstacles in robotic navigation. Real-life geometric structures, such as floor plans~\cite{schwarzkpf1996} and street maps~\cite{CDGNW2011} are often low-density environments, and many robotic navigation problems can be solved more efficient if the environment is assumed to be low-density~\cite{VANDERSTAPPEN1993}. 

Chen et al.~\cite{CDGNW2011} combined low density and packedness to produce a faster map-matching algorithm. Given a polygonal curve $\pi$ and a graph $G$ with edges embedded as straight line segments, they considered the problem of matching $\pi$ to a path $P$ in $G$ that minimises the Fréchet distance between $\pi$ and $P$. Assuming that $G$ is $\lambda$-low-density, and $\pi$ is $c$-packed, their algorithm runs in near-linear time. Furthermore, they verified that the maps of San Francisco, Athens, Berlin, and San Antonio are all $\lambda$-low-density where $\lambda$ is a small constant. 

Although we can produce fast algorithms by assuming that the data set fits in one of the realistic input models, how can we check if these assumptions are reasonable? Researchers either assume that the majority of a particular type of data set fits into one of the models (floor plans~\cite{schwarzkpf1996}, and obstacles~\cite{VANDERSTAPPEN1993}) or use slow methods to verify existing data sets (English handwriting~\cite{driemel2012}, and city maps~\cite{CDGNW2011}). Such methods may not be suitable given the proliferation of geometric data that is nowadays generated, and often in a dynamic setting. 

Therefore, in this paper, we study the problem of deciding the $\lambda$-low-density of a given set of segments in the Euclidean plane. 

To the best of our knowledge, the algorithm by De Berg et al.~\cite{deBerg2002} is the current state-of-the-art for deciding the density value of a set of objects, and takes $O(n \log^3n + \lambda n \log^2n + \lambda^2 n)$ time~\cite{deBerg2002}. However, their algorithm can only handle the restricted case when segments do not intersect. Furthermore, no data structure is known that allows for dynamic update on the $\lambda$-low-density value.

The main result of this paper is: Given a set of $n$ line segments in $\mathbb{R}^2$ one can compute a $3$-approximation of the $\lambda$-low density value in $O(n \log n + \lambda n)$ time. We also show how to maintain a $3$-approximation of the $\lambda$-low density value while allowing insertions of new segments in $O(\log n + \lambda^2)$ amortized time per update.


To investigate the usefulness of the $\lambda$-low density model for trajectories we implemented a $4$-approximation algorithm to estimate the density values of twelve real-world data sets. The median density values in eleven of the data sets are less than $51$. The median $\lambda / n$ ratios of six data sets are less than $0.04$, where $n$ is the size of the curve. Although there are only twelve data sets and our values are estimates, the results indicate that low density is a practical, realistic model for many real-world data sets.

Due to the space constraint, some of the proofs can be found in the appendix.

\section{Approximation algorithms} \label{cha:algorithms} 
In this section, we show a simple quadratic time 25-approximation algorithm to compute the density value of a given set of segments, which will then be improved upon in later sections. We start by formally defining low density for segments (see also Figure~\ref{fig:low_density_segments}).

\begin{figure}[!htb]
    \centering
    \includegraphics[scale=0.9]{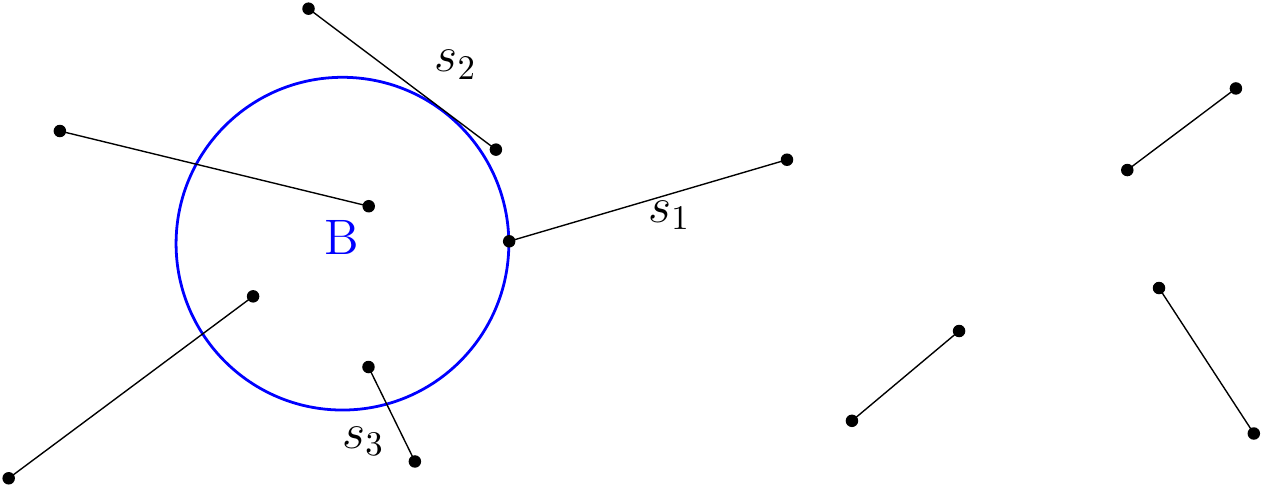}
    \caption{The set of segments in the figure is $4$-low-density. Segments $s_1$ and $s_2$ touches $B$, while $s_3$ does not count as intersecting $B$ since $|s_3| < \text{radius}(B)$. The ball $B$ is an optimal ball as it intersects four segments, and no ball can intersect five or more segments.}
    \label{fig:low_density_segments}
\end{figure}

\begin{definition}
\label{def:low_density}
Let $S := \{s_1, ..., s_n\}$ be a set of segments, and let $\lambda \geq 0$ be a parameter. We say $S$ is $\lambda$-\textit{low-density} if for any ball $B$, the number of segments $s_i \in S$ with $|s_i| \geq$ \text{radius}($B$) that intersect $B$ is at most $\lambda$.
\end{definition}

To simplify the description of the algorithms and the proofs we say a segment $s$ \emph{intersects} a ball $B$ if and only if the length of $|s|\geq \text{radius}($B$)$ and $s\cap B \neq \emptyset$, and $s$ is said to \emph{touch} a ball $B$ if $s$ intersects $B$ at a single point. Finally, a ball $B$ is optimal with respect to a set $S$ of segments if $S$ is $\lambda$-low-density and $B$ intersects exactly $\lambda$ segments of $S$.


\subsection{Basic properties and a first algorithm} \label{ssec:properties}
In this section we will prove some basic properties that will lead us to a simple approximation algorithm. 

\begin{lemma} \label{lemma:same_size}
Let a set $S$ of segments in $\mathbb{R}^2$ be $\lambda$-low-density. There exists an optimal ball $B$ and a segment $s$ intersecting $B$ such that $|s| = \text{radius}(B)$, and $s$ is the shortest segment that intersect $B$.
\end{lemma}
\begin{proof}
Let $B$ be an optimal ball that intersect with $\lambda$ segments. Let $s$ be the shortest segment that intersects $B$. If $|s| = \text{radius}(B)$, we are done. If $\abs{s} > \text{radius}(B)$, then let $B'$ be a ball with the same center as $B$, and let $\text{radius}(B') = |s|$. Then $B'$ still intersects $\lambda$ segments. If $\abs{s} < \text{radius}(B)$, then $s$ does not intersect $B$ which contradicts the assumption.
\end{proof}

Given Lemma~\ref{lemma:same_size}, if $s$ is the shortest segment that intersect an optimal ball $B$, then it immediately follows that $B$ must lie entirely within a stadium-shaped region $P_s$ around $s$. That is, given a segment $s$, let $P_s$ be the union of all points within distance $2\abs{s}$ of a point on $s$, as shown in Figure~\ref{fig:stadium}(a).

\begin{figure}[!htb]
    \centering
    \includegraphics[width=0.9\textwidth]{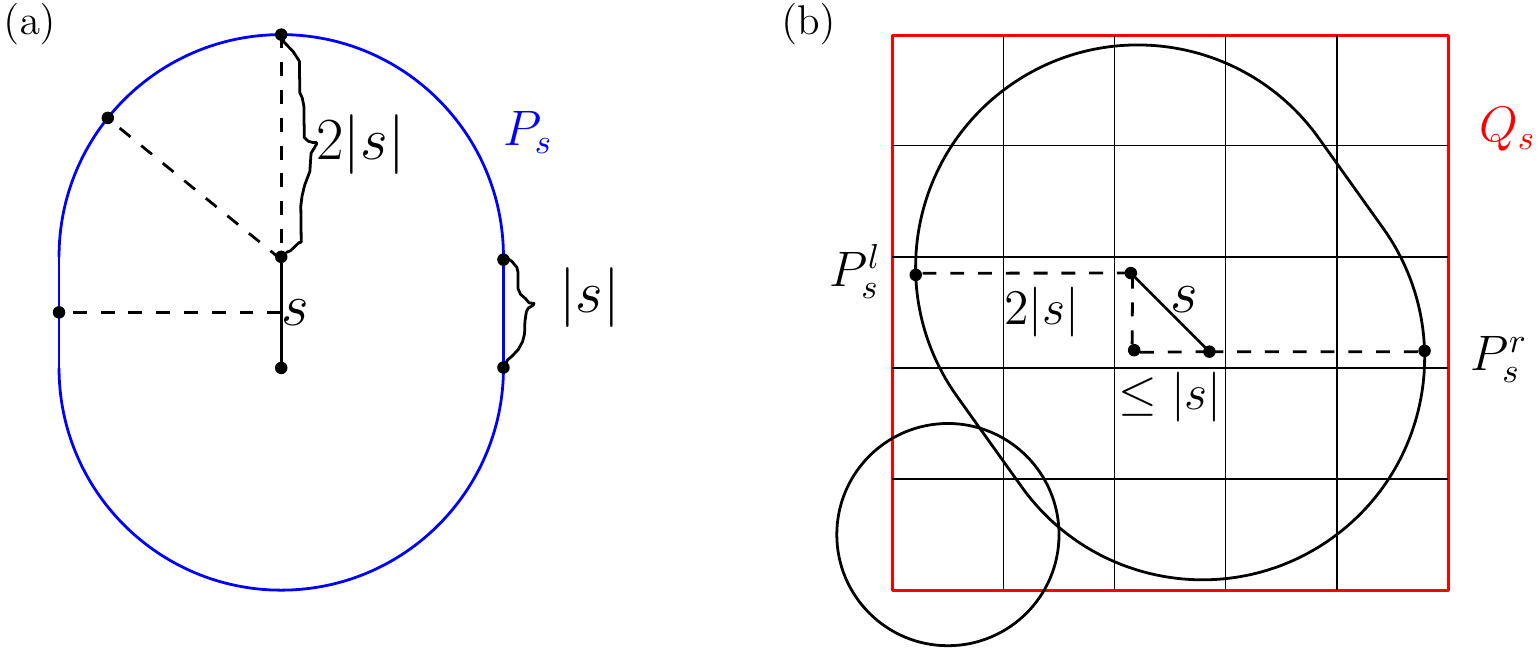}
    \caption{(a) A segment $s$ and its stadium $P_s$. (b) $Q_s$ covers $P_s$.}
    \label{fig:stadium}
\end{figure}

Next, we show that $P_s$ can be covered by a $5\abs{s}$ by $5\abs{s}$ axis-aligned square $Q_s$ centered at the middle point of $s$, and $Q_s$ can be partitioned into 25 smaller squares of side length~$\abs{s}$. Each of these smaller squares can be covered by a ball of radius $\abs{s}$. As a result, if there exists an optimal ball $B$ intersecting $s$ and having radius $|s|$ then there exists a small ball that intersects at least $\lambda / 25$ segments.

The above arguments suggests a natural $O(n^2)$ algorithm. By generating $25$ balls per segment and calculating the maximum number of segments that intersect any ball, we can compute a $25$-approximation of the density value of a set of segments. The approach is outlined in Algorithm~\ref{alg:25-approximation}. 

\begin{algorithm}
\caption{A 25-approximation algorithm}\label{alg:25-approximation}
\begin{algorithmic}[1]
\Require $S$: a set of segments in $\mathbb{R}^2$
\State $\lambda \gets 0$
\ForEach {$s \in S$}
    \State $S' \gets$ the set of segments that intersect $Q_s$
    \State $balls \gets 25$ balls that covers $Q_s$
    \ForEach {$\iscircle$ in $balls$}
        \State $\lambda \gets$ max($\lambda$, the number of segments in $S'$ that intersect $\iscircle$)
    \EndFor
\EndFor
\State \Return $\lambda$
\end{algorithmic}
\end{algorithm}

From the above arguments, we can also bound the number of segments that intersect~$Q_s$ which will be used in Section~\ref{sec:quadtree} to speed up the algorithm. 

\begin{corollary}
\label{lemma:o_lambda}
Let $S$ be a set of $\lambda$-low-density segments. For any $s \in S$, the number of segments $s' \in S$ with $\abs{s'} \geq \abs{s}$ that intersect $Q_s$ is $O(\lambda)$.
\end{corollary}


\subsection{Improving the approximation factor to three}
\label{sec:34-approximation}
In the previous section we gave a simple approximation algorithm with a rough approximation factor for ease of explanation. The main idea of the $25$-approximation algorithm is to cover the stadium $P_s$ with a square grid containing $25$ small squares, and then cover each square with a ball. This guarantees that there exists one ball that intersects at least $\lambda/25$ segments.

In this section we will instead use a sheared triangular grid $G_T$ of equilateral triangles having side length $|s|/9$. The sheared grid will have side length $10|s| \times 10|s|$ and will be centered at the middle point of a segment $s$. Note that the sheared grid covers the stadium~$P_s$, as illustrated in Fig.~\ref{fig:rhombus}.

\begin{figure}[!htb]
    \centering
    \includegraphics[width=0.7\textwidth]{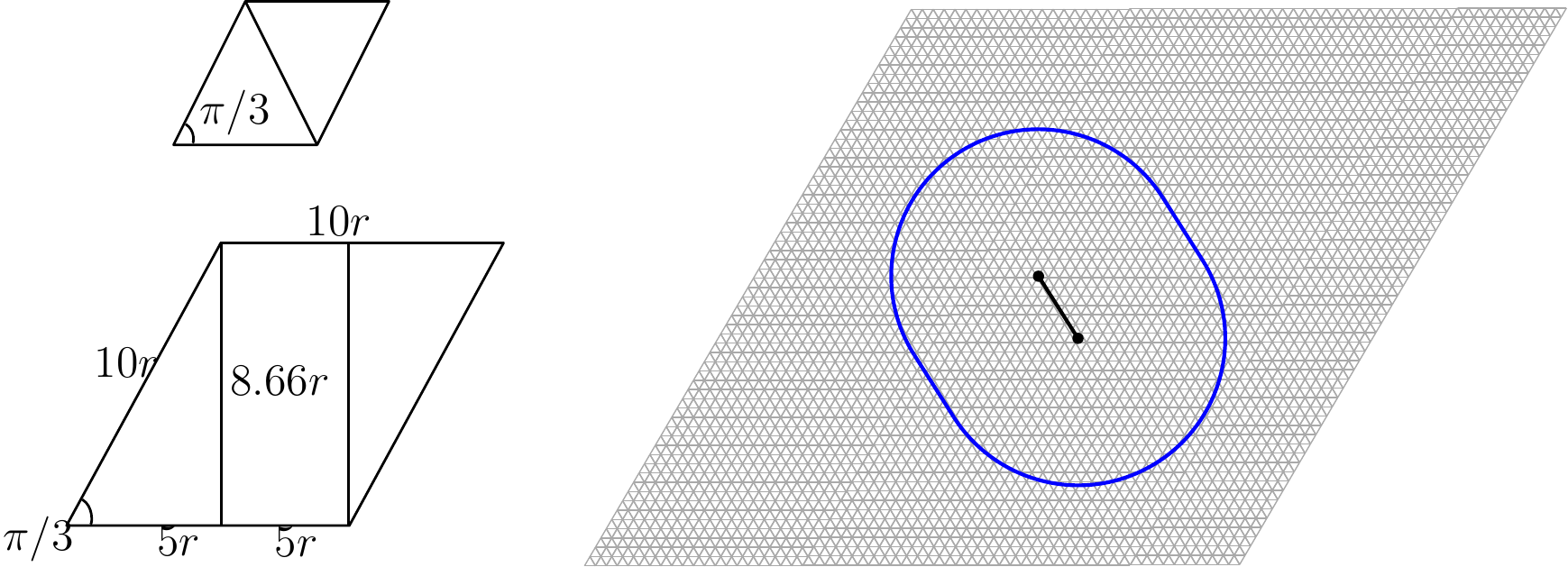}
    \caption{Illustrating the sheared $10|s| \times 10|s|$ triangular grid $G_T$ of equilateral triangles having side length $|s|/9$.}
    \label{fig:rhombus}
\end{figure}


We will argue that for every triangle $T$ in the grid one can construct three balls of radius $|s|$ such that any optimal ball of radius $|s|$ and center within $T$, can be covered by three three balls. This will reduce the approximation factor to $3$, while only increasing the running time by a constant factor (due to increasing the number of balls we need to check). We formally define the set of balls that cover $Q_s$. For an arbitrary triangle $T$ in $G_T$, without loss of generality, assume that the bottom side of $T$ aligns with the $x$-axis, and let $(x, y)$ be $T$'s center. We place three balls of radius $\abs{s}$ centered at $(x - (\frac{\sqrt{3}}{4} + \frac{1}{36})|s|, y - (\frac{1}{4} + \frac{1}{36\sqrt{3}})|s|)$, $(x, y - (\frac{1}{2} + \frac{1}{18 \sqrt{3}})|s|)$, and $(x + (\frac{\sqrt{3}}{4} + \frac{1}{36})|s|, y + (\frac{1}{4} + \frac{1}{36\sqrt{3}})|s|)$, respectively. We denote the set of balls constructed from the triangular grid for $s$ with $\mathcal{B}_s$.

We first describe, in Observation~$\ref{obs:minkowski_sum}$, the region where an optimal ball can exist. Let $P, Q \subseteq \mathbb{R}^d$. We denote Minkowski sum~\cite{Minkowski_sum} by $P \oplus Q = \{p + q \mid p \in P, q \in Q\}$. Note the following observation:

\begin{observation}
\label{obs:minkowski_sum}
Let $Q \subseteq \mathbb{R}^d$, $Q \neq \emptyset$, let $B_0$ be a ball of radius $|s|$ centered at the origin, and let $B$ be a ball of radius $|s|$. If $\text{center}(B) \in Q$, then $B \subseteq Q \oplus B_0$. 
\end{observation}

Using Observation~\ref{obs:minkowski_sum} we can improve the approximation factor to three. It suffices to to show that the Minkowski sum of a region and $B_0$ can be covered by three balls in $\mathcal{B}_s$. 

\begin{lemma}
\label{lemma:3-approx}
Let $T$ be an equilateral triangle in $G_T$, and let $B_0$ be a ball with radius $|s|$ centered at the origin. There exists three balls in $\mathcal{B}_s$ whose union cover $T \oplus B_0$.
\end{lemma}
\begin{proof}
Consider an optimal ball $B$ of radius $|s|$, and let $T$ be the equilateral triangle containing the centre of $B$. Without loss of generality, assume that the bottom side of $T$ aligns with $x$-axis and let $(x,y)$ be the center of $T$. 

From the definition of the set $\mathcal{B}_s$ we know there exists three balls $A$, $B$, and $C$ in $\mathcal{B}_s$ centered at $(x - (\frac{\sqrt{3}}{4} + \frac{1}{36})|s|, y - (\frac{1}{4} + \frac{1}{36\sqrt{3}})|s|)$, $(x, y - (\frac{1}{2} + \frac{1}{18 \sqrt{3}})|s|)$, and $(x + (\frac{\sqrt{3}}{4} + \frac{1}{36})|s|, y + (\frac{1}{4} + \frac{1}{36\sqrt{3}})|s|)$, respectively. Each of these balls have radius $|s|$.

\begin{figure}[!htb]
    \centering
    \includegraphics[width=\textwidth]{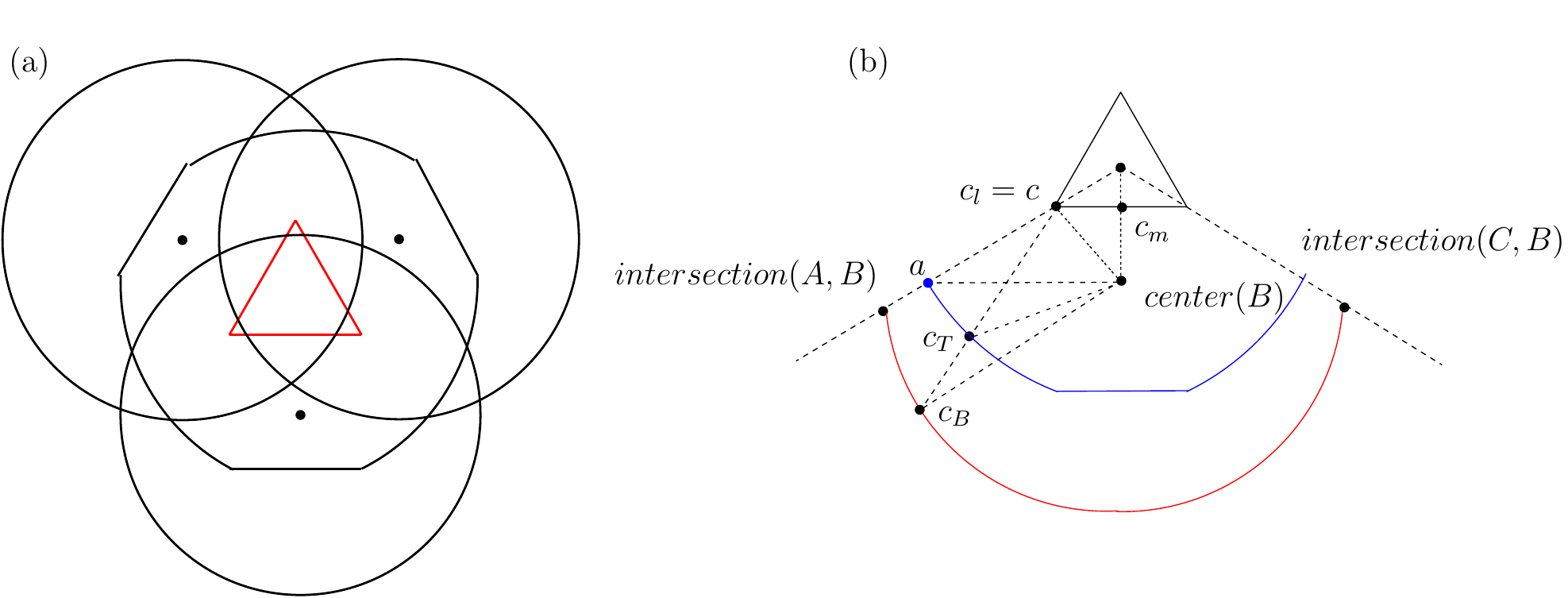}
    \caption{(a) Three balls cover $T \oplus B_0$ (notice that in actual construction, $T$ is much smaller compare to the balls.) (b) The red line shows the border of $B$. The blue line shows the border of the bottom one-third sector of $T \oplus B_0$. }
    \label{fig:3_triangle_cover}
\end{figure}

Let the bottom-left corner of $T$ be $c_l$, let $a$ be the point on the border of $T \oplus B_0$ such that $a$ is on the line extended from $(\text{center}(T), c_l)$. Due to construction, $\text{dist}(c_l, a) = |s|$. We will prove that $T \oplus B_0 \subseteq A \cup B \cup C$ by showing that $B$ covers the bottom one-third sector of $T \oplus B_0$. For now, we will focus on the bottom one-third sector of $T \oplus B_0$, see Figure~\ref{fig:3_triangle_cover}(b).

Let $c_m$ be the middle point of the bottom side of $T$. We focus on the points on $(c_l, c_m)$, and their nearest points on the border of $T \oplus B_0$. Let $c \subseteq (c_l, c_m)$, and let $c_T$ be the nearest point of $c$ on the border of $T \oplus B_0$. Let $c_B$ be the point on the border of $B$ such that $c_B$ lies on the line extended from $(c, c_T)$. We show that $B$ covers half of the one-third sector by showing that $\text{dist}(c_T, c_B) > 0$. 

Notice that $\text{intersection}(A, B)$ is guaranteed to lie on the line extended from $(c_l, a)$ since $\text{dist}(\text{center}(B), \text{center}(T)) = \text{dist}(\text{center}(A), \text{center}(T))$. 

\begin{align*}
    \measuredangle(a, \text{center}(T), \text{center}(B)) = cos^{-1}(\frac{\frac{1}{2} + \frac{1}{18 \sqrt{3}}}{1 + \frac{1}{9 \sqrt{3}}}) = \frac{\pi}{3} \\
    \measuredangle(\text{center}(T), a, \text{center}(B)) = sin^{-1}(\frac{\frac{1}{2} + \frac{1}{18 \sqrt{3}}}{1 + \frac{1}{9 \sqrt{3}}}) = \frac{\pi}{6} \\
    \measuredangle(a, \text{center}(B), \text{center}(T)) \\
    = \pi -  \measuredangle(a, \text{center}(B), \text{center}(T)) - \measuredangle(\text{center}(T), a, \text{center}(B)) 
    = \frac{\pi}{2} 
\end{align*}

Then we can calculate $\text{dist}(a, \text{center}(B))$:
\begin{align*}
    \text{dist}(a, \text{center}(B)) = cos(\frac{\pi}{6}) \cdot (1 + \frac{1}{9 \sqrt{3}}) |s| \approx 0.922 |s| < |s| 
\end{align*}

We have $\text{dist}(c, a) \geq \text{dist}(c, \text{intersection}(A, B))$ because $\text{dist}(a, \text{center}(B)) < \text{radius}(B)$, and since $\text{dist}(\text{intersection}(A, B), \text{center}(B)) = |s|$, we know that $\text{intersection}(A, B)$ lies below and to the left of $a$. Therefore $\text{dist}(a, \text{intersection}(A, B)) > 0$. 

Let $c \subseteq (c, c_m)$. We focus on the triangle created by $c$, $c_T$ and $\text{center}(B)$. As we slide $c_T$ along the border of $T \oplus B_0$ towards its bottom, $\text{dist}(\text{center}(B), c_T)$ decreases which means $\text{dist}(\text{center}(B), c_T) < 0.922|s| < |s| = \text{dist}(\text{center}(B), c_B)$. Therefore $\text{dist}(c_T, c_B) > 0$ for $c \subseteq (c_l, c_m)$. Let $c_r$ be the bottom-right corner of $T$, we can arrive at the same conclusion for $c \subseteq (c_m, c_r)$. Therefore $B$ covers the bottom one-third sector of $T \oplus B_0$. 

Using a similar argument, we can prove that the $A$ and $C$ cover the top-left, and top-right one-third sectors of $T \oplus B_0$, respectively. Therefore $T \oplus B_0 \subseteq A \cup B \cup C$, which completes the proof of the lemma.
\end{proof}

This proves that any optimal ball intersecting $s$ can be covered by three of the balls in $\mathcal{B}_s$. Using Algorithm~\ref{alg:25-approximation}, but replacing the square grid with the sheared triangular grid described in this section, we obtain the following theorem.

\begin{theorem}
\label{thm:3-approx}
 A $3$-approximation of the $\lambda$-low-density value of a set $S$ of $n$ segments in the plane can be computed in $O(n^2)$ time.
\end{theorem}


\section{Improving the running time} \label{sec:quadtree}
The bottleneck of of Algorithm~\ref{alg:25-approximation} is step~3, which for each segment $s \in S$ constructs the set $S' \subseteq S$ containing all the segments in $S$ longer than $s$ and intersecting $Q_s$. For each segment this takes linear time. However, from Corollary~\ref{lemma:o_lambda}, we know that $S'$ only contains $O(\lambda)$ segments and we also know the regions $Q_s$ for all segments $s \in S$ at the start of the algorithm. We will utilise both these observations to speed up Step 3 of the algorithm by modifying a compressed quadtree~\cite{geometric_approximation}. We start by introducing some preliminaries, and prove a key lemma that will allow us to store the segments in a compressed quadtree.

\subsection{Preliminaries}
Scale and translate the input such that all endpoints of the segments are contained in the unit square $[0,1] \times [0,1]$.

During the construction process of the quadtree, we continuously divide each square into four equal-sized squares. Each cell is also a node in the quadtree. The $x$ and $y$-coordinates of the corners of these cells have the form $k / 2^i$, where $i$ and $k$ are integers. In a grid $G_r$, a grid cell has side length $r$, see Figure~\ref{fig:grid}.

\begin{figure}[!htb]
    \centering
    \includegraphics[width=\textwidth]{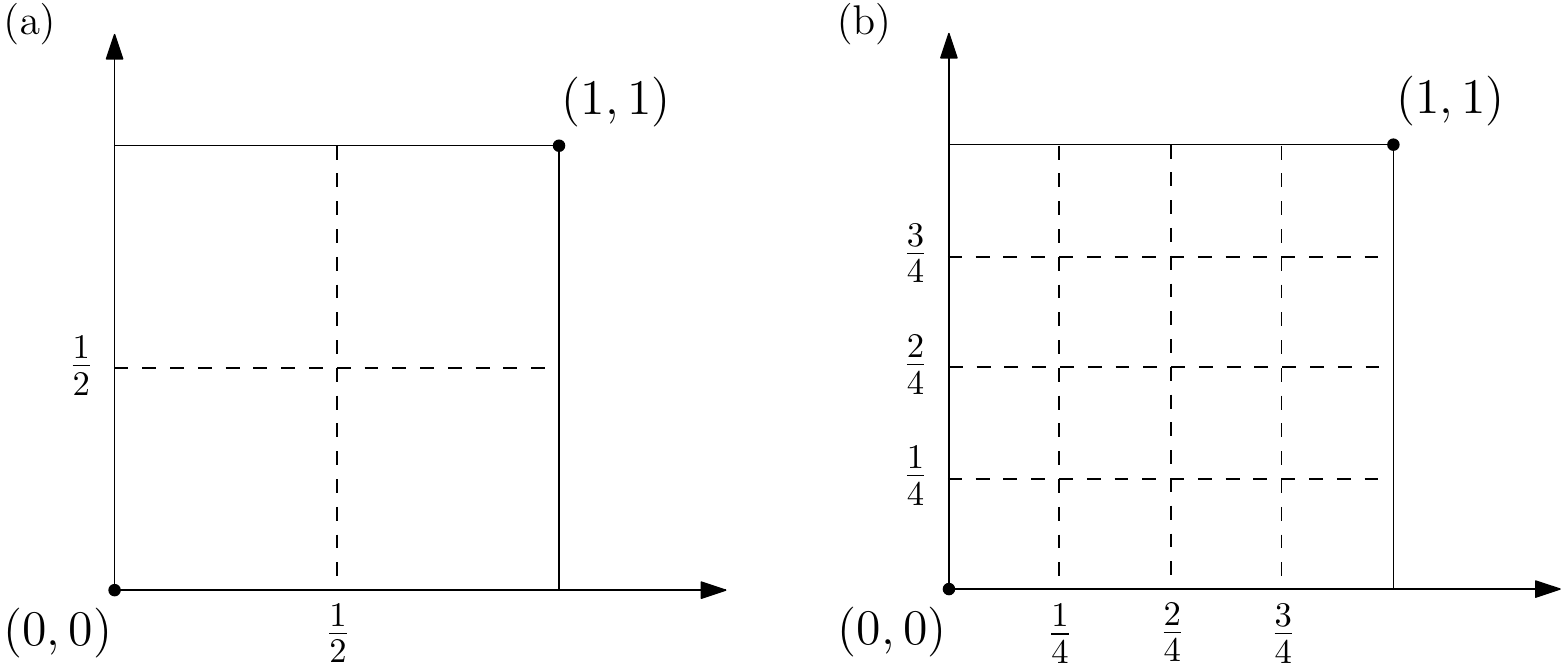}
    \caption{The grid described by the squares of size $1/2^i$ is denoted $G_{1/2^i}$. (a) Showing $G_{1/2}$, and (b) $G_{1/4}$.}
    \label{fig:grid}
\end{figure}

In the quadtree construction, the root of the quadtree is always the unit square. We say a square $\square$ is a \textit{canonical square} if and only if $\square$ is a square in the constructed quadtree. 

With unit square and canonical square defined, we will next show that there exists a constant number of canonical squares that cover $Q_s$. These canonical squares can then be used to construct a quadtree to store our segments.

In order to cover $Q_s$ with canonical squares, we first need to define the \emph{most significant separator} of $Q_s$. Consider the vertical line at $x = \frac{k}{2^i}$, $i, k \in \mathbb{Z}$, where $k$ is odd. We say that it is a vertical \textit{separator} of $Q_s$ if it intersects $Q_s$. We say $x = \frac{k}{2^i}$ is more significant than $x = \frac{k'}{2^j}$ if $i < j$. We define a horizontal separator symmetrically. Notice that each $Q_s$ can intersect exactly one most significant vertical and one most significant horizontal separator. Indeed, if $Q_s$ intersects two most significant vertical separators $x = \frac{k}{2^i}$ and $x = \frac{k'}{2^i}$, $k \neq k'$, there exits some $k'' \in (k, k')$ such that $k''$ is even, since both $k$ and $k'$ are odd. Then $Q_s$ also intersects $x = \frac{k''}{2^i} = \frac{k'' / 2}{2^{i + 1}}$, which means $x = \frac{k}{2^i}$ cannot be the most significant separator. 

The proof of the following lemma can be found in Appendix~\ref{sec:appendix_proof_lemma}.
\begin{lemma}
\label{lemma:split_constant}
For all $s \in S$, we can construct $O(1)$ canonical squares to cover $Q_s$ such that the side length of each canonical square is at most $\abs{s}$. 
\end{lemma}

In the proof of Lemma~\ref{lemma:split_constant} we used a constructive proof to show how to cover $Q_s$ with canonical squares. However, to compute the canonical squares efficiently, we need to find the most significant separator of $Q_s$ fast. Next we show that computing the most significant separator can be done in constant time, assuming the unit RAM model. 

\begin{lemma}
\label{lemma:find_most_significant_separator}
Finding the most significant separators of $Q_s$ takes $O(1)$ time under the unit RAM model.
\end{lemma}

\begin{proof}
We use the \textit{bit twiddling} method in \cite{geometric_approximation}. We assume that $\alpha, \beta \in [0, 1)$ are written in base two as $\alpha = 0.\alpha_1 \alpha_2$..., and $\beta = 0. \beta_1 \beta_2$... Let $\text{bit}_\triangle (\alpha, \beta)$ be the index of the first bit after the period in which they differ.

Let $\alpha, \beta$ be the $x$-coordinate of the top-left and top-right corner of $Q_s$, respectively. Let $i = \text{bit}_\triangle (\alpha, \beta)$. Notice that if the first $i - 1$ bits of $\alpha$ and $\beta$ are the same, they must reside in a canonical square which is a cell of $G_{1 / 2^{i - 1}}$. For example, if $\alpha_1 = \beta_1 = 1$, then $\alpha, \beta \in [1 / 2, 1)$. Furthermore, if $\alpha_2 = \beta_2 = 0$, then $\alpha, \beta \in [2 / 4, 3 / 4)$. Therefore, the most significant vertical separator must be either $x = \alpha$, $x = \beta$ or the immediate next separator to the right side of $\alpha$ which takes the form $x = k / 2^{i}$, where $k$ is odd. 

We can then find the most significant separator by adding $1 / 2^i$ to $\alpha$, and set the bits after $i$th bit to $0$. Let $l = 0.\alpha_1 \alpha_2 ... \alpha_{2^{i - 1}} 1$. We report the more significant separator among $\alpha$, $\beta$, and $l$ as the most significant vertical separator of $Q_s$. And since it takes $O(1)$ time to calculate $\text{bit}_\triangle (\alpha, \beta)$, calculating $l$ and comparing $l$, $\alpha$, and $\beta$ also takes $O(1)$ time.
\end{proof}

Har-Peled~\cite{geometric_approximation} justified why it is reasonable to assume that $\text{bit}_\triangle(\cdot, \cdot)$ can be computed in constant time. Modern float number is represented by exponent and mantissa. If two numbers have different exponents, then we can compute $\text{bit}_\triangle(\cdot, \cdot)$ by giving the larger component. Otherwise, we can XOR the mantissas, and $\log_2(\cdot)$ the result. Har-Peled~\cite{geometric_approximation} also noted that these are built-in operations on some CPU models. Therefore assuming that $\text{bit}_\triangle(\cdot, \cdot)$ requires constant time to compute is a reasonable assumption.

\subsection{Preprocessing: Generate canonical squares} \label{ssec:preprocessing}
Above we proved that we can generate a set $\mathcal{C}_s$ of canonical squares of size at most $|s|$ that cover~$Q_s$. Given a set $S$ of segments and the set $\mathcal{C}_s$ for each segment $s$ in $S$, we call $s$ the \textit{associated segment} of the squares in $\mathcal{C}_s$, and we call the squares in $\mathcal{C}_s$ the \textit{associated squares} of~$s$. We perform the following steps:  

\begin{enumerate}
    \item For each segments $s$, perform a linear scan of the canonical squares in $\mathcal{C}_s$ and mark the squares that intersect~$s$. We call $s$ the \textit{intersecting segment} of these squares, and we call these squares the \textit{intersecting squares} of~$s$. 
    \item Let $\mathcal{C}=\cup_{s\in S} \mathcal{C}_s$ and sort all the canonical squares in $\mathcal{C}$ based on their sizes, then by coordinates of the bottom-left corners, and, finally, by the length of the segment they intersect in decreasing order.
    \item To remove duplicate squares in $\mathcal{C}$, perform a linear scan of the sorted squares in the list, merge adjacent squares if they are identical. While merging two squares, also merge the set of segments they associate with and the set of segments they intersect. The segments that intersect the same square are ordered in decreasing order of length.
\end{enumerate}

At the end of the preprocessing step, we have $O(n)$ canonical squares according to Lemma~\ref{lemma:split_constant}. Each segment keeps a list of pointers to its canonical squares, and each canonical square holds the segment(s) it is associated with, and the segment(s) it intersects in decreasing order of length. Let~$Q_s'$ be the region covered by the set of canonical squares in $\mathcal{C}_s$. We can now summarise the preprocessing step.

\begin{lemma}
\label{lemma:preprocessing_complexity}
The preprocessing step generates $O(n)$ canonical squares, where each canonical square stores pointers to the segments intersecting it, in decreasing order of length, and its associated segments. The preprocessing step takes $O(n \log n)$ time.
\end{lemma}
\begin{proof}
Scaling and translating segments within the unit square takes $O(n)$ time. Then for each segment $s$, cover $Q_s$ in $O(1)$ time by using the method described in the proofs of Lemmas~\ref{lemma:split_constant} and~\ref{lemma:find_most_significant_separator}. In Lemma~\ref{lemma:split_constant}, we showed that each $Q_s$ generates a constant number of canonical squares, therefore this step takes $O(n)$ time, and generates $O(n)$ canonical squares. Removing duplicates requires a sorting step, which takes $O(n\log n)$ time. 
\end{proof}

\subsection{An efficient construction using compressed quadtrees} \label{ssec:quadtree}
We will construct the compressed quadtree from the $O(n)$ canonical squares constructed in the previous section. Let $\square_v$ denote the square associated with node $v$. We say a segment $s$ intersects a node $v$ in the compressed quadtree if $s$ intersects $\square_v$. We use the below lemma from~\cite{geometric_approximation}. 

\begin{lemma}[Lemma~2.11 in~\cite{geometric_approximation}] 
\label{clemma:quadtree_from_squares}
Given a list $\mathcal{C}$ of $n$ canonical squares, all lying inside the unit square, one can construct a (minimal) compressed quadtree $\mathcal{T}$ such that for any square $c \in C$, there exists a node $v \in \mathcal{T}$, such that $\square_v = c$. The construction time is $O(n \log n)$.
\end{lemma}

We now apply Lemma~\ref{clemma:quadtree_from_squares}, with the $O(n)$ canonical squares constructed in Section~\ref{ssec:preprocessing} as input, we obtain a compressed quadtree, where each canonical square corresponds to an internal node. 

Once we complete the construction of the compressed quadtree $\mathcal{T}$, we start the \emph{push-down step} as follows. Sort the segments in increasing order of length, then iterate through the segments. For each segment $s \in S$, go to the internal nodes in $\mathcal{T}$ that it intersects with. Search all of their children to check if $s$ intersects with them. If $s$ does, insert $s$ at the beginning of the list of intersecting segments of that internal node and continue the search in its children.  

Using a naive analysis, since the compressed quadtree has at most $n$ nodes, applying the push-down step for all $n$ segments takes $O(n^2)$ time in total. However, below we will show that the running time improves when the analysis is done in terms of the $\lambda$-low-density value of $S$ as well as $n$.

First note that the arguments for Corollary~\ref{lemma:o_lambda} can easily be extended to prove the following corollary.

\begin{corollary}
\label{cor:o_lambda}
Let $S$ be a set of $\lambda$-low-density segments in $\mathbb{R}^2$. For any $s \in S$, let $Q$ be a $a \abs{s} \times b \abs{s}$ rectangle such that $a$ and $b$ are constants, and $P_s \subseteq Q$. The number of segments $s' \in S$ with $\abs{s'} \geq \abs{s}$ that intersect $Q$ is $O(\lambda)$.
\end{corollary}

The corollary implies that only $O(\lambda)$ segments intersects the canonical squares covering~$Q_s$. Using this result we can now improve the running time.

\begin{lemma}
\label{lemma:quadtree_construction}
The compressed quadtree takes $O(n\log n + \lambda n)$ time to construct. The push-down step takes $O(\lambda n)$, and the resulting data structure uses $O(\lambda n)$ space.
\end{lemma}
\begin{proof}
According to Lemma~\ref{lemma:preprocessing_complexity} and Lemma~\ref{clemma:quadtree_from_squares}, one can construct a compressed quadtree with $O(n)$ squares in $O(n \log n)$ time that contains the $O(n)$ required canonical squares. It remains only to analyse the running time of the push-down step. 

Sorting the segments takes $O(n\log n)$ time. Recall that in  Corollary~\ref{lemma:o_lambda}, we showed that for any $s \in S$, the number of segments that are longer than or equal to the length of $s$ that intersect with $Q_s$ is $O(\lambda)$. The same holds for $Q_s'$. Notice that in the push-down step, for each associated internal node $v$ of a segment $s$, one only add $s$ to the descendants of $v$ if they intersect $s$. The side lengths of associated canonical squares of a segment have to be less than or equal to the length of that segment. Therefore, $s$ is only added to the internal nodes whose associated segments are shorter than or equals to $s$. In another word, an internal node $v$ adds $s$ as intersecting segment if and only if $\abs{s}$ is greater than or equal to the length of the associated segment of $v$. As a result each internal node adds at most $\lambda$ segments. And since each $Q_s'$ is covered by a constant number of squares, we add $O(\lambda)$ segments to each $Q_s'$. The push-down step takes $O(\lambda n)$ time. The overall time complexity is $O(n \log n + \lambda n$).

We have a constant number of associated squares for each segment. Since each square results in at most one internal nodes and four children, we have a compressed quadtree with $O(n)$ nodes. In addition, each internal node associated with each $Q_s'$ stores $O(\lambda)$ intersecting segments. Therefore we stored $O(\lambda n)$ intersecting segments and the compressed quadtree uses $O(\lambda n)$ space in total. 
\end{proof}

\subsection{The final approximation algorithm}
Given the $3$-approximation algorithm in Section~\ref{sec:34-approximation} together with the compressed quadtree introduced in Section~\ref{ssec:quadtree}, we are now ready to present the final algorithm. In Algorithm~\ref{alg:quadtree_3-approx}, we query the compressed quadtree described in the previous section to get our $3$-approximation. Recall that $Q_s'$ is the region covered by the set of canonical squares we generated from $Q_s$ in the preprocessing step. Similarly to what we have done in Section~\ref{sec:34-approximation}, we will use a set of triangles to cover $Q_s'$, and we will cover each triangle with three balls. We denote the set of balls that cover $Q_s'$ by $\mathcal{B}_s'$.

\begin{algorithm}[!htb]
\caption{A faster 3-approximation using quadtrees} \label{alg:quadtree_3-approx}
\begin{algorithmic}[1]
\Require $S$: a set of segments in $\mathbb{R}^2$ 
\State $\mathcal{T} \gets$ a compressed quadtree storing $S$ (as described in Section~\ref{ssec:quadtree}.)
\State $\lambda \gets 0$
\ForEach{$s$ in $S$}

        \State $S' \gets$ intersecting segments of the associated nodes of $s$ of length at least $|s|$
        \State $\mathcal{B}_s' \gets$ the set of balls of radius $\abs{s}$ that cover $Q_s'$. 
        
            \ForEach{$\iscircle$ in $\mathcal{B}_s'$}
                \State $\lambda$ = max($\lambda$, number of segments in $S'$ that intersect $\iscircle$)
            \EndFor
\EndFor

\State \Return $\lambda$
\end{algorithmic}
\end{algorithm}

\begin{theorem}
\label{thm:quadtree_3-approx}
Given a set $S$ of segments in $\mathbb{R}^2$, one can compute a $3$-approximation of the density value of $S$ in $O(n \log n + \lambda n)$ time. 
\end{theorem}
\begin{proof}
Iterate through all the segments, and one of them must be the shortest segment $s$ that intersects an optimal ball. According to Lemma~\ref{lemma:split_constant}, the union of the associated squares of $s$ covers $Q_s$, and each associated square has side length less than or equals to $\abs{s}$. In Section~\ref{sec:34-approximation} it was shown that one can use a constant number of equilateral triangles of side length $\abs{s} / 9$ to cover a $Q_s$, therefore the same can be done to each $Q_s'$. In Lemma~\ref{lemma:3-approx}, we have shown that there exists three balls in $\mathcal{B}_s$ of radius $\abs{s}$ that cover any optimal ball centered anywhere in each of the triangles. Therefore our algorithm generates a constant number of balls that cover $Q_s'$, the region covered by the associated squares of $s$, and the ball that intersects maximum number of segments intersects at most $\lambda / 3$ segments.

We have shown in Lemma~\ref{lemma:quadtree_construction} that building a quadtree with $n$ segments takes $O(n \log n + \lambda n)$ time. For each segment, one can generate $O(1)$ canonical squares (Lemma~\ref{lemma:split_constant}), therefore there are $O(1)$ associated nodes for each segment $s$. We build the quadtree during the construction so that the intersecting segments of duplicated canonical squares are sorted in decreasing order of length. Notice that during the push-down step, we preserve this order by pushing down the shorter segments first, then the longer ones. Due to Corollary~\ref{cor:o_lambda}, the number of segments with lengths longer than or equals to $\abs{s}$ that intersect $Q_s'$ is $O(\lambda)$, hence, it takes $O(\lambda)$ time to retrieve the segments $S'$ in $S$ that intersects with $Q_s'$, and since we generates a constant number of balls, it takes $O(\lambda)$ time to find the ball that intersects the most number of segments. Summing up over all the $n$ segments the final running time of this step is $O(\lambda n + n \log n)$. 

The time complexity of this algorithm is dominated by the construction of the compressed quadtree, which is $O(n \log n + \lambda n)$.
\end{proof}

This completes the main result of the paper. In the appendix we show how one can modify the construction of the quadtree using a $\mathcal{Z}$-order of the canonical squares, as proposed by Har-Peled~\cite{geometric_approximation}. This gives the following theorem:
\begin{theorem} \label{thm:insertion}
One can maintain a $3$-approximation of the density value of set of segments in $O(\log n + \lambda^2$) amortized time per insertion. 
\end{theorem}

\section{Experiments}
\label{cha:experiments}
We implemented a simple algorithm to obtain the $4$-approximate density values on several trajectory data sets. The algorithm is a simplification of the 3-approximation algorithm presented in Section~\ref{sec:34-approximation}.

The data sets were provided to us by the authors of~\cite{pfeifer_data}. Each data set contains a number of trajectories. Table~\ref{table:data_summary} summarises the data sets, and its information is taken from~\cite{pfeifer_data}. This experiment aims to show the practicality of our approximation method and motivate the study of low density.  

\begin{table}[!bth]
\centering
\caption{Real data sets containing trajectories in $\mathbb{R}^2$, showing number of input trajectories $n$, average number of vertices per trajectory, and a description}
\label{table:data_summary} 
\renewcommand{\arraystretch}{1.2}
\begin{tabular}{ l  l  l  l  l } 
\hline
Data set & $n$ & \#Vertices & Trajectory Description \\ [0.5ex] 
\hline
Vessel-M \cite{data_vessel} & 106  & 23.0 & MS River USA shipping vessels Shipboard AIS.  \\
Vessel-Y \cite{data_vessel} & 187  & 155.2 & Yangtze River shipping vessels Shipboard AIS.\\
Truck \cite{data_bus_truck} & 276  & 406.5 & GPS of 50 concrete trucks in Athens, Greece.\\
Bus \cite{data_bus_truck} & 148  & 446.6 & GPS of School buses.\\
Taxi \cite{data_taxi1, data_taxi2} & 180,736 & 75.7 & Beijing taxi trajectories split into trips. \\ 
Geolife \cite{data_geolife1, data_geolife2, data_geolife3} & 18,670 & 1,332.5 & People movement, mostly in Beijing, China.\\ [0.5ex] 
\hline 
Pigeon \cite{data_pigeon} & 131  & 970.0 & Homing Pigeons (release sites to home site). \\
Seabird \cite{data_seabird} & 134& 3,175.8  & GPS of Masked Boobies in Gulf of Mexico.\\
Cats \cite{data_cat} & 154 & 526.1  & Pet house cats GPS in RDU, NC, USA. \\

Buffalo \cite{data_buffalo} & 165  & 161.3  & Radio-collared Kruger Buffalo, South Africa. \\
Gulls \cite{data_gulls} & 253 & 602.1 & Black-backed gulls GPS (Finland to Africa).\\
Bats \cite{data_bats} & 545  & 127.2 & Video-grammetry of Daubenton trawling bats.\\

\hline 
\end{tabular}
\end{table}

We used the 4-approximation algorithm to estimate the density values of the trajectories in each data set. We record the total number of curves, the maximum curve size, the max, and median estimate density values, and the median density to curve size ratio for each data set. We summarise the information in Table~\ref{table:result_lambda}. Additional information and experiments can be found in Appendix~\ref{app:experiments}.

\begin{table}[!htb]
\centering
\caption{The table lists $4$-approximate density values of $12$ data sets. The second and third columns show the number of curves and the maximum curve size. The following two columns shows the maximum and median density value. The last column shows the median ratio between $\lambda$ and the size of the curve $n$. }
\label{table:result_lambda}
\renewcommand{\arraystretch}{1.2}
\begin{tabular}{ l  l  l  l  l  c } 
\hline
Data set & \#Curves & MaxCurveSize & Max & Median & Median $\lambda / n$ \\ [0.5ex] 
\hline
Vessel-M & 102 & 142 & 17 & 2.0 & 0.133 \\
Vessel-Y & 186 & 559  & 3 & 3.0 & 0.020 \\
Truck & 272 & 991  & 45 & 11.0 & 0.028 \\
Bus & 144 & 1015  & 17 & 7.0 & 0.016 \\
Taxi & 947 & 3340  & 682 & 14 & 0.091 \\
GeoLife & 999 & 64482  & 335 & 5 & 0.011 \\
\hline
Pigeon & 130 & 1645  & 665 & 28.0 & 0.038 \\
Seabird & 133 & 8556  & 1483 & 351 & 0.131 \\
Cat & 153 & 11122  & 411 & 45 & 0.224 \\
Buffalo & 162 & 479  & 82 & 21.0 & 0.169 \\
Gull & 126 & 16019  & 1520 & 50.5 & 0.159 \\
Bat & 544 & 735  & 8 & 2.0 & 0.020 \\

\hline 
\end{tabular}
\end{table}

In the experiment, we estimated the density values of trajectories in twelve real-world data sets. Our observation is that most of the data sets' estimated density values are low. In eleven data sets (all except Seabird), the median estimate density values are less than $51$. In six data sets, Vessel-Y, Truck, Bus, GeoLife, Pigeon, and Bat, the median $\lambda / n$ ratios are less than $0.04$, where $\lambda$ and $n$ are the estimated density value and the size of the curve. Although the density values are estimates, the notion of low density for trajectories is valuable. The algorithms that perform better when the density value is small, e.g. the algorithm by Driemel et al.~\cite{driemel2012}, can be applied to these data sets to improve efficiency.

\section{Concluding remarks} \label{sec:results_algorithm}
In this paper we considered the problem of approximating the $\lambda$-density value of a set of segments. Our main results is a $3$-approximation algorithm running in  $O(n \log n + \lambda n)$ time, where $\lambda$ is the density value. Previously, only an $O(n \log^3 n + \lambda n \log^2 n + \lambda^2 n)$ time algorithms was known for the special case when the segments are disjoint. Our approach can be extended to handle insertions, where each insert operation can be done in $O(\log n + \lambda^2)$ amortized time.  

We also implemented a simple $4$-approximation algorithm to estimate the density values of twelve real-world trajectory data sets. We observed that the estimated densities for most of the data sets are small constants. We also observed that the trajectories in half of the data sets have low density-to-size ratios, which indicates that low density is a practical, realistic input model for trajectories.



\bibliography{lipics-v2019-sample-article}

\appendix
\section{Proof of Lemma~\ref{lemma:split_constant}}
\label{sec:appendix_proof_lemma}
\begin{customlemma}{\ref{lemma:split_constant}}
For all $s \in S$, we can construct $O(1)$ canonical squares to cover $Q_s$ such that the side length of each canonical square is at most $\abs{s}$. 
\end{customlemma}

\begin{figure}[!htb]
    \centering
    \includegraphics[width=\textwidth]{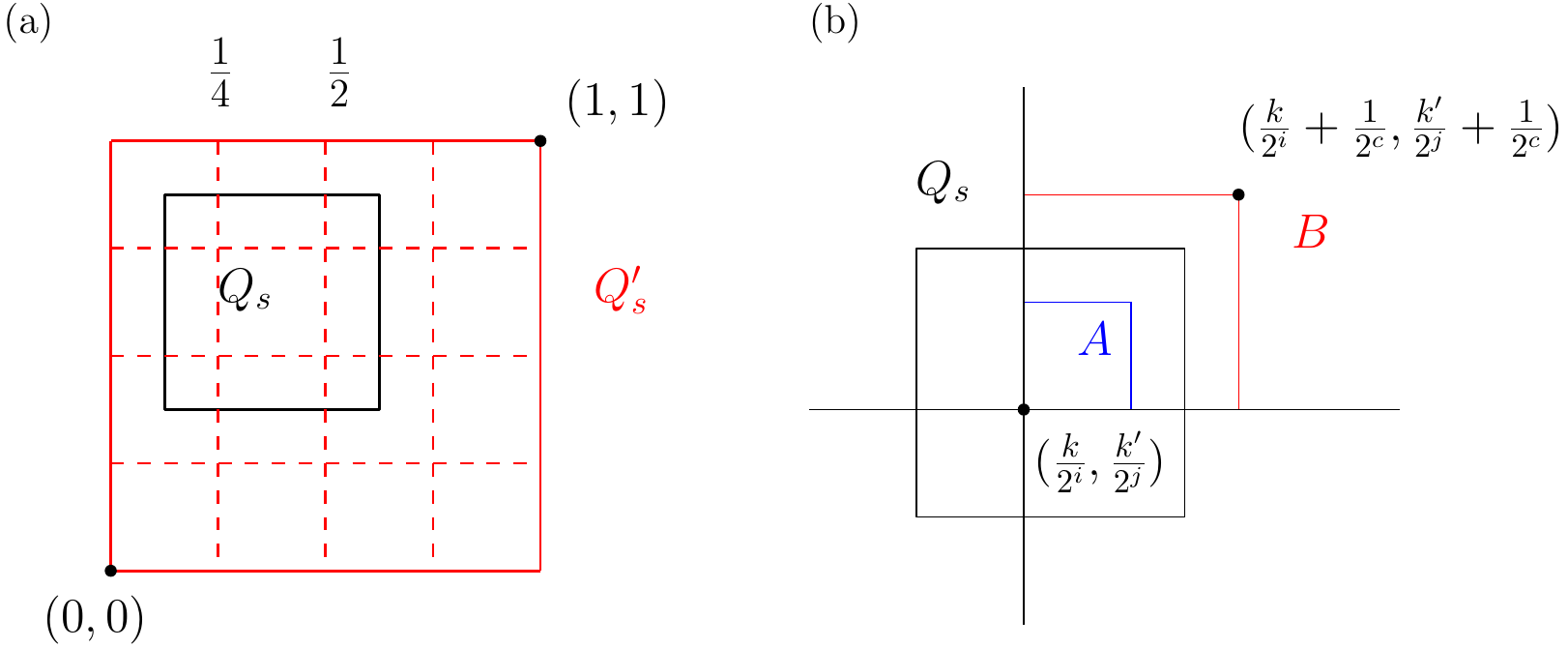}
    \caption{(a) Expanding $Q_s$ into $Q_s'$ so it can be split up into canonical squares. We say $x = \frac{1}{2}$ is a more significant separator than $x = \frac{1}{4}$. (b) Expanding the top-right part of $Q_s$ into $B$, whose side is twice as long as the side of $A$.}
    \label{fig:expand_Qs}
\end{figure}

\begin{proof}
If $Q_s$ is identical to a canonical square then we are trivially done. If not, for an arbitrary $Q_s$, let $x = \frac{k}{2^i}$ and $y = \frac{k'}{2^j}$ be the most significant vertical and horizontal separators, respectively. They intersect at $(\frac{k}{2^i}, \frac{k'}{2^j})$, and separate $Q_s$ into four parts (see Figure~\ref{fig:expand_Qs}(b)). Without loss of generality, let top-right part be the largest part. Let $(a, b)$ be the top-right corner of $Q_s$, and let $c = -\lceil \log_2(\text{max}(\abs{a - \frac{k}{2^i}}, \abs{b - \frac{k'}{2^j}}) \rceil$. 

We expand the top-right part of $Q_s$ to a square $B$ which has $(\frac{k}{2^i}, \frac{k'}{2^j})$ as bottom-left corner, and $(\frac{k}{2^i} + \frac{1}{2^c}, \frac{k'}{2^j} + \frac{1}{2^c})$ as top-right corner (non-inclusive). Notice that $c > i, j$, because the opposite suggests that $Q_s$ intersects a more significant vertical separator than $x$ or a more significant horizontal separator than $y$. 
Therefore $B$ is the intersection of the following halfplanes.
\begin{align*}
    x &\geq \frac{k}{2^i} = \frac{1}{2^c} \cdot \frac{k}{2^{i - c}} \\
    y &\geq \frac{k'}{2^j} = \frac{1}{2^c} \cdot \frac{k'}{2^{j - c}} \\
    x &< \frac{k}{2^i} + \frac{1}{2^c} = \frac{1}{2^c} \cdot (\frac{k}{2^{i-c}} + 1) \\
    y &< \frac{k'}{2^j} + \frac{1}{2^c} = \frac{1}{2^c} \cdot (\frac{k'}{2^{j - c}} + 1) 
\end{align*}
Since $c > i, j$, and $c, i, j, k, k' \in \mathbb{Z}$, both $\frac{k}{2^{i - c}}$ and $\frac{k'}{2^{j - c}}$ are integers, which means $B$ satisfies the definition of a cell in $G_{1 / 2^c}$. And since $B$ is contained inside the unit square, $B$ is a canonical square. 

We now show that $B$ can be split into a constant number of canonical squares with side lengths less than or equal to $\abs{s}$. Let $d = -\lfloor \log_2(\text{max}(\abs{a - \frac{k}{2^i}}, \abs{b - \frac{k'}{2^j}}) \rfloor$. Let $A$ be the square with $(\frac{k}{2^i}, \frac{k'}{2^j})$ as its bottom-left corner, and $(\frac{k}{2^i} + \frac{1}{2^d}, \frac{k'}{2^j} + \frac{1}{2^d})$ as its top-right corner. $Q_s$ has to be bigger than $A$ as the opposite suggests $A = B$. Therefore we assume $Q_s$ is bigger than $A$. The side length of $B$ is less than two times the side length of $Q_s$, assuming that the area of the other three parts of $Q_s$ is minimal. Therefore the side length of $B$ is less than $10 \abs{s}$. Notice that dividing a canonical square evenly into four squares also gives canonical squares. Therefore we can divide $B$ into canonical squares. Since the side length of $B$ is less than $10 \abs{s}$, it takes at most $\lceil \log_2(10) \rceil = 4$ splits before the resulting canonical squares have side lengths less than $\abs{s}$. We can expand the rest of the quadrants to the same size as $B$, and split them using the same method. 

The above argument applies to the scenario where $Q_s' = Q_s$ is already a canonical square. A constant number of splits for each of four parts of $Q_s'$ leads to a constant number of canonical squares. The proof is complete.
\end{proof}

\section{Make the compressed quadtrees dynamic}
A natural follow-up question of the preceding section is whether our compressed quadtree can be made dynamic. Insert and delete operations are helpful as they allow dynamic input. 

Har-Peled~\cite{geometric_approximation} defined an ordering of points and canonical squares, which we will call the $\mathcal{Z}$-order. As in previous sections, we are only interested in storing canonical squares. Therefore we only give an overview of the definition of $\mathcal{Z}$-order for canonical squares.

First, imagine we have a quadtree, and we traverse this quadtree with the DFS traversal (see Figure~\ref{fig:q_order}). And on each root of a subtree, we always traverse its children in a set order. We always visit the bottom-left child first, then the bottom-right child, top-left child, and top-right child. Then the DFS traversal defines a total ordering of the canonical squares. We say $\square \prec \widehat{\square}$ if we visit $\square$ before we visit $\widehat{\square}$ using the above DFS traversal.

\begin{figure}[!htb]
    \centering
    \includegraphics[width=\textwidth]{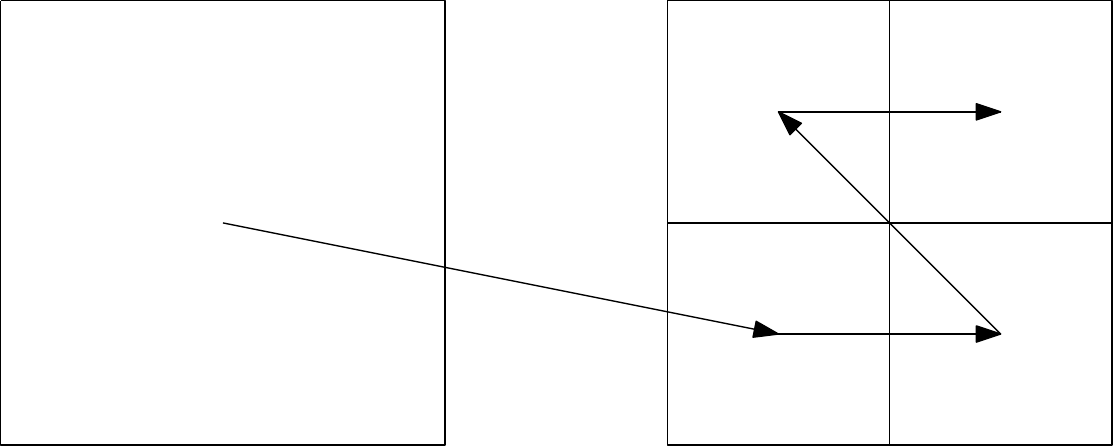}
    \caption{The $\mathcal{Z}$-order. Right canonical squares are the children of left canonical square.}
    \label{fig:q_order}
\end{figure}

With the $\mathcal{Z}$-order defined, we can store the canonical squares in a balanced binary search data structure such as a skip list or an AVL tree. However, we need to be able to resolve the $\mathcal{Z}$-order of two canonical squares in $O(1)$ time. 

Observe that in a quadtree, distinct canonical squares have distinct centers, and $\text{center}(\square)$ lies in $\widehat{\square}$ if and only if $\square \subseteq \widehat{\square}$. Furthermore, if $\square \subseteq \widehat{\square}$, then $\widehat{\square} \prec \square$, and if $\widehat{\square} \subseteq \square$, then $\square \prec \widehat{\square}$. Therefore we can decide if one square is the ancestor of another, and if so, resolve the $\mathcal{Z}$-order of two canonical squares in $O(1)$ time. 

Otherwise, we need to find the smallest canonical square that contains both $\square$ and $\widehat{\square}$ by using $\text{bit}_\triangle(\cdot, \cdot)$ operation. Recall that we denoted $\text{bit}_\triangle(\alpha, \beta)$ as the first  bit after the period in which they differ, where $\alpha = 0.\alpha_1 \alpha_2 ...$, and $\beta = 0.\beta_1 \beta_2 ...$. By using $\text{bit}_\triangle(\cdot, \cdot)$ on the centers of $\square$ and $\widehat{\square}$, we can obtain the level then length of their smallest common canonical square $\widetilde{\square}$. Once we get the length, we can cast the center of $\square$ to the bottom-left corner of $\widetilde{\square}$ to obtain $\widetilde{\square}$. The above description was summarised in the below corollary from \cite{geometric_approximation}.

\begin{customcorollary}{2.17}[from~\cite{geometric_approximation}]
Assuming that the $\text{bit}_\triangle$ operation and the $\lfloor \cdot \rfloor$ operation can be performed in constant time, then one can compute the LCA (smallest canonical square) of two points (or cells) in constant time. Similarly, their $\mathcal{Z}$-order can be resolved in constant time.
\end{customcorollary}

Being able to resolve the ordering of two canonical squares in $O(1)$ time, one can store these canonical squares in a sorted, balanced binary data structure. Therefore we have the below theorem from \cite{geometric_approximation}.

\begin{customthm}{2.22}[from~\cite{geometric_approximation}]
\label{cthm:quadtree_operations}
Assuming one can compute the $\mathcal{Z}$-order (of two points or cells) in constant time, then one can maintain a compressed quadtree of a set of points in $O(\log n)$ time per operation, where insertions, deletions, and point-location query are supported. Furthermore, this can be implemented using any data structure for ordered-set that supports an operation (insert, delete, and point-location query) in logarithmic time. 
\end{customthm}

Our compressed quadtree can handle insert operation using Algorithm~\ref{alg:insert}. The invariants maintained by the algorithm are the density value and the fact that each internal node stores its intersecting and associated segments.

\begin{algorithm}[!htb]
\caption{Insert} \label{alg:insert}
\begin{algorithmic}[1]

\Require $s$: a segment in $\mathbb{R}^2$; $qt$: a compressed quadtree; $\lambda$: the $3$-approximation of the density value of segments in $qt$.
\State $squares$ $\gets$ the canonical squares split from $Q_s'$
\State Insert $squares$ into $qt$ (using Theorem~\ref{cthm:quadtree_operations}), and merge two squares if they exist in $qt$.
\ForEach{$v$ in the inserted nodes}
    \State $v_p \gets$ the parent of $v$
    \State Add intersecting segments of $v_p$ that are longer than or equals to $\abs{s}$ as intersecting segments of $v$
    \ForEach {descendant $d$ of $v$}
        \State insert $s$ as the intersecting segments of $d$ if $\square_d$ intersects $s$
        \State $\lambda \gets$ max($\lambda$, 3-approximation($\square_v$, $S'$)) 
    \EndFor
\EndFor

\State \Return $\lambda$

\end{algorithmic}
\end{algorithm}

\begin{theorem}
Insert operation, Algorithm~\ref{alg:insert},  takes worst case $O(\lambda n )$ time, and amortized $O(\log n + \lambda^2$) time per segment, and can be used to update the $3$-approximation of the $\lambda$-low-density value in the same time. 
\end{theorem}
\begin{proof}
The correctness of inserting a canonical square into a compressed quadtree is shown in Theorem~\ref{cthm:quadtree_operations} from \cite{geometric_approximation}. 

We focus on proving that the insert operation correctly maintains the lists of intersecting segments of $\mathcal{T}$. We know that a segment $s$ can only intersect $\square$ if $s$ intersects a canonical square containing $\square$. Therefore for each inserted node $v$, every segment that intersects $\square_v$ must be part of the intersecting segments of $\square_{v_p}$. Thus we only need to visit $v_p$ to retrieve the intersecting segments of $v$. And by using the same reasoning, only the descendants of $v$ can intersect $s$, and we can preserve the decreasing order of the intersecting segments of $v$ by using a binary insertion. 

According to Theorem~\ref{cthm:quadtree_operations} from~\cite{geometric_approximation}, inserting $O(1)$ canonical squares into a quadtree takes $O(\log n)$ time. Recall that in Corollary~\ref{cor:o_lambda}, we showed that the number of segments that intersect any bounding box with side length at most a constant times the length of $s$ is $O(\lambda)$. Therefore each node of $\mathcal{T}$ contains $O(\lambda)$ intersecting segments. Thus, for each inserted node $v$, iterating over the intersecting segments of $v_p$ takes $O(\lambda)$ time. And inserting $s$ as the intersecting segments of any node takes $O(\log \lambda)$ using binary insertion. 

What remains is to upper-bound how many nodes we need to insert $s$ into. In the worst case, we may need to insert $s$ into $O(n)$ nodes. Running our $3$-approximation algorithm on $n$ nodes takes $O(\lambda n)$ time. Therefore the insert operation takes worst case $O(\log n + n \log \lambda + \lambda n) \subseteq O(\lambda n)$ time. 

However, we can argue that on average, we only update $O(\lambda)$ nodes per insert. Since each $Q_s'$ intersects $O(\lambda)$ segments, after $n$ consecutive insert operations, we insert only $O(\lambda n)$ segments as the intersecting segments of various nodes. Therefore we insert $s$ into $O(\lambda)$ nodes on average. The amortized time of each insert operation is $O(\log n + \lambda \log \lambda + \lambda^2) \subseteq O(\log n + \lambda^2)$. This completes the proof.

\end{proof}

\section{Appendix: Experiments}\label{app:experiments}
We implemented a simple algorithm to obtain the $4$-approximate density values on several trajectory data sets. The algorithm is a simplification of the 3-approximation algorithm presented in Section~\ref{sec:34-approximation}.

The data sets were provided to us by the authors of~\cite{pfeifer_data}. Each data set contains a number of trajectories. Table~\ref{table:data_summary2} summarises the data sets, and its information is taken from~\cite{pfeifer_data}. This experiment aims to show the practicality of our approximation method and motivate the study of low density. By estimating the density values of real-world trajectories, we hope to show that many curves are low-density; therefore, the notion of low density is a practical, realistic input model. 

\begin{table}[!htb]
\centering
\caption{Real data sets containing trajectories in $\mathbb{R}^2$, showing number of input trajectories $n$, average number of vertices per trajectory, and a description}
\label{table:data_summary2}
\renewcommand{\arraystretch}{1.2}
\begin{tabular}{ l  l  l  l  l } 
\hline
Data set & $n$ & \#Vertices & Trajectory Description \\ [0.5ex] 
\hline
Vessel-M \cite{data_vessel} & 106  & 23.0 & MS River USA shipping vessels Shipboard AIS.  \\
Vessel-Y \cite{data_vessel} & 187  & 155.2 & Yangtze River shipping vessels Shipboard AIS.\\
Truck \cite{data_bus_truck} & 276  & 406.5 & GPS of 50 concrete trucks in Athens, Greece.\\
Bus \cite{data_bus_truck} & 148  & 446.6 & GPS of School buses.\\
Taxi \cite{data_taxi1, data_taxi2} & 180,736 & 75.7 & Beijing taxi trajectories split into trips. \\ 
Geolife \cite{data_geolife1, data_geolife2, data_geolife3} & 18,670 & 1,332.5 & People movement, mostly in Beijing, China.\\ [0.5ex] 
\hline 
Pigeon \cite{data_pigeon} & 131  & 970.0 & Homing Pigeons (release sites to home site). \\
Seabird \cite{data_seabird} & 134& 3,175.8  & GPS of Masked Boobies in Gulf of Mexico.\\
Cats \cite{data_cat} & 154 & 526.1  & Pet house cats GPS in RDU, NC, USA. \\

Buffalo \cite{data_buffalo} & 165  & 161.3  & Radio-collared Kruger Buffalo, South Africa. \\
Gulls \cite{data_gulls} & 253 & 602.1 & Black-backed gulls GPS (Finland to Africa).\\
Bats \cite{data_bats} & 545  & 127.2 & Video-grammetry of Daubenton trawling bats.\\

\hline 
\end{tabular}
\end{table}

\subsection{Experimental setup}

We used the 4-approximation algorithm to estimate the density values of the trajectories in each data set. We record the total number of curves, the maximum curve size, the max, and median estimate density values, and the median density to curve size ratio for each data set. We summarise the information in Table~\ref{app:table:result_lambda}. We present the distributions of estimate density values in Figure~\ref{fig:distribution1} and \ref{fig:distribution2}. 

\begin{table}[!htb]
\centering
\caption{The table lists $4$-approximate density values of $12$ data sets. The second and third columns show the number of curves and the maximum curve size. The following two columns shows the maximum and median density value. The last column shows the median ratio between $\lambda$ and the size of the curve $n$. }
\label{app:table:result_lambda}
\renewcommand{\arraystretch}{1.2}
\begin{tabular}{ l  l  l  l  l  c } 
\hline
Data set & \#Curves & MaxCurveSize & Max & Median & Median $\lambda / n$ \\ [0.5ex] 
\hline
Vessel-M & 102 & 142 & 17 & 2.0 & 0.133 \\
Vessel-Y & 186 & 559  & 3 & 3.0 & 0.020 \\
Truck & 272 & 991  & 45 & 11.0 & 0.028 \\
Bus & 144 & 1015  & 17 & 7.0 & 0.016 \\
Taxi & 947 & 3340  & 682 & 14 & 0.091 \\
GeoLife & 999 & 64482  & 335 & 5 & 0.011 \\
\hline
Pigeon & 130 & 1645  & 665 & 28.0 & 0.038 \\
Seabird & 133 & 8556  & 1483 & 351 & 0.131 \\
Cat & 153 & 11122  & 411 & 45 & 0.224 \\
Buffalo & 162 & 479  & 82 & 21.0 & 0.169 \\
Gull & 126 & 16019  & 1520 & 50.5 & 0.159 \\
Bat & 544 & 735  & 8 & 2.0 & 0.020 \\

\hline 
\end{tabular}
\end{table}

\begin{figure}[!htb]

\centering
    \begin{minipage}{0.5\textwidth}
        \centering
        \includegraphics[width=\textwidth]{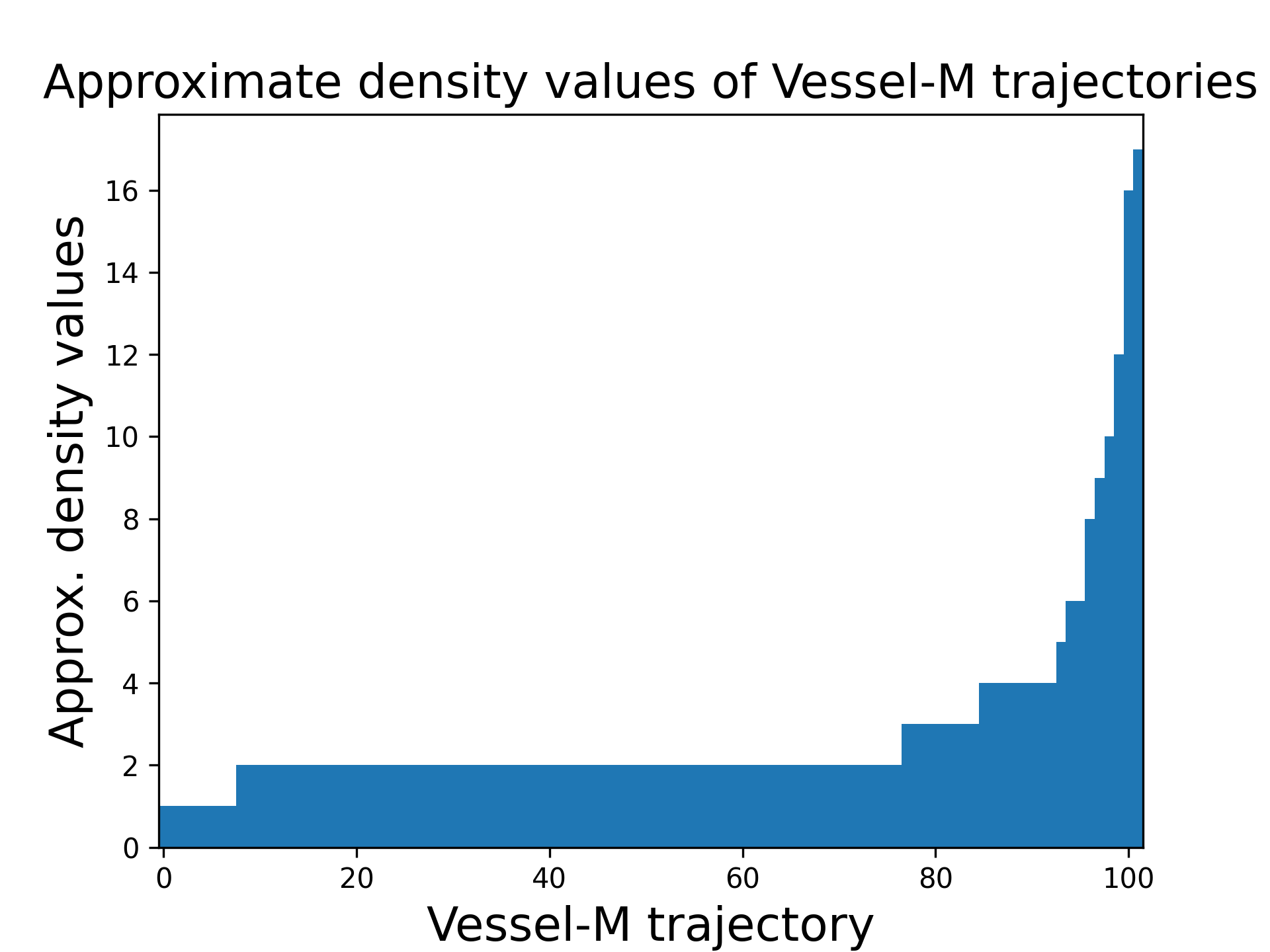} 
    \end{minipage}\hfill
    \begin{minipage}{0.5\textwidth}
        \centering
        \includegraphics[width=\textwidth]{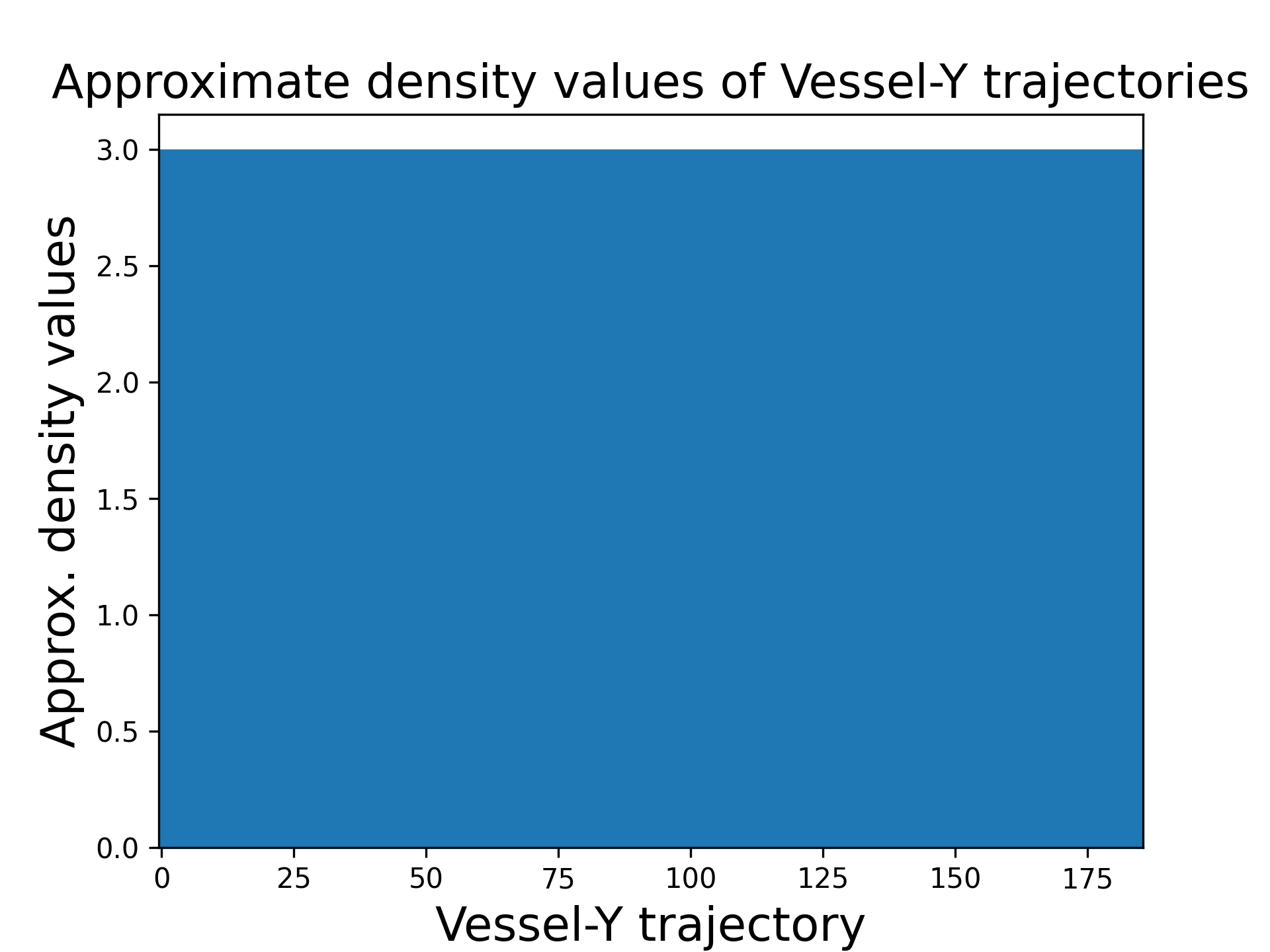} 
    \end{minipage}
    
    \begin{minipage}{0.5\textwidth}
        \centering
        \includegraphics[width=\textwidth]{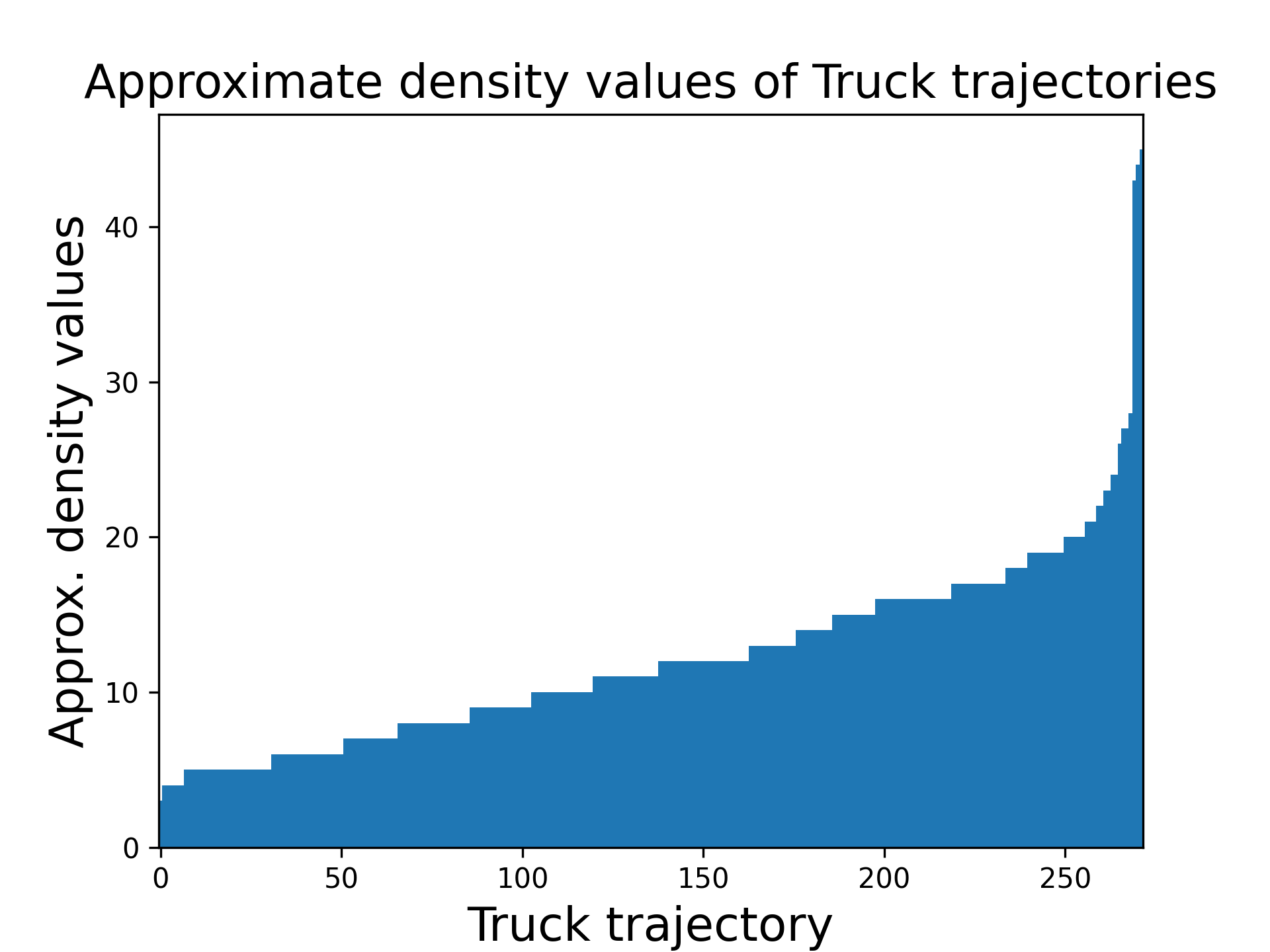} 
    \end{minipage}\hfill
    \begin{minipage}{0.5\textwidth}
        \centering
        \includegraphics[width=\textwidth]{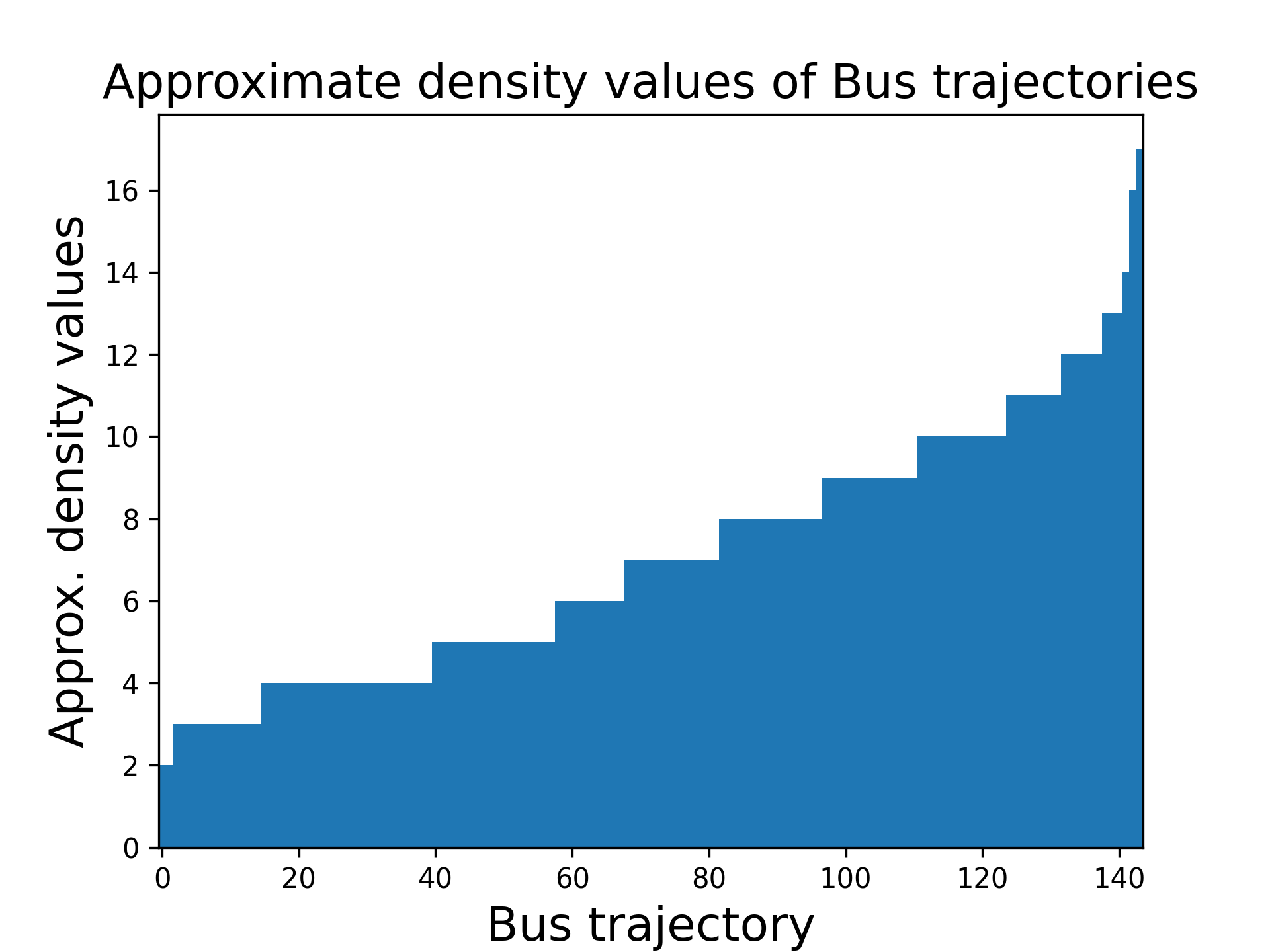} 
    \end{minipage}
    
    \begin{minipage}{0.5\textwidth}
        \centering
        \includegraphics[width=\textwidth]{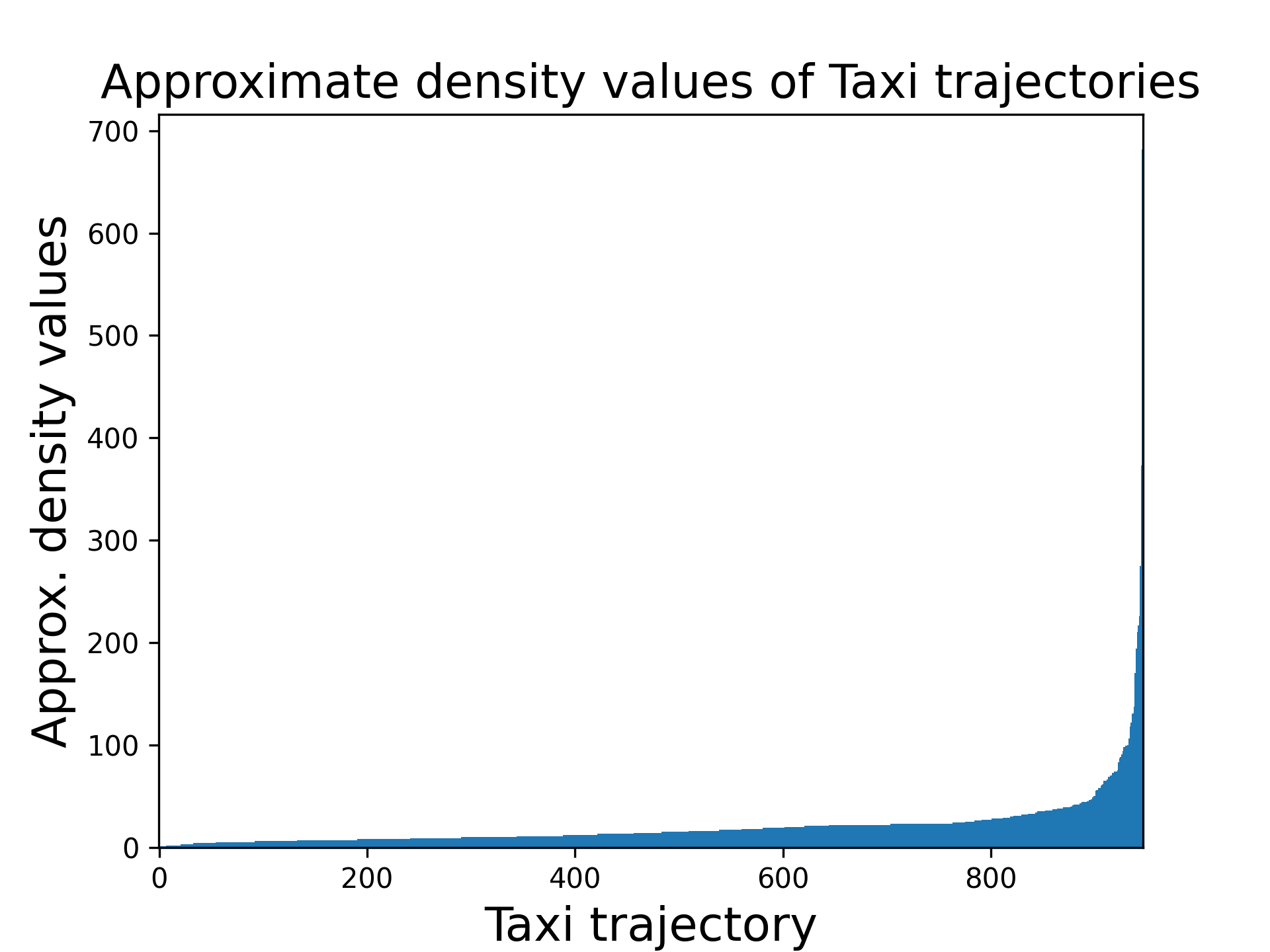} 
    \end{minipage}\hfill
    \begin{minipage}{0.5\textwidth}
        \centering
        \includegraphics[width=\textwidth]{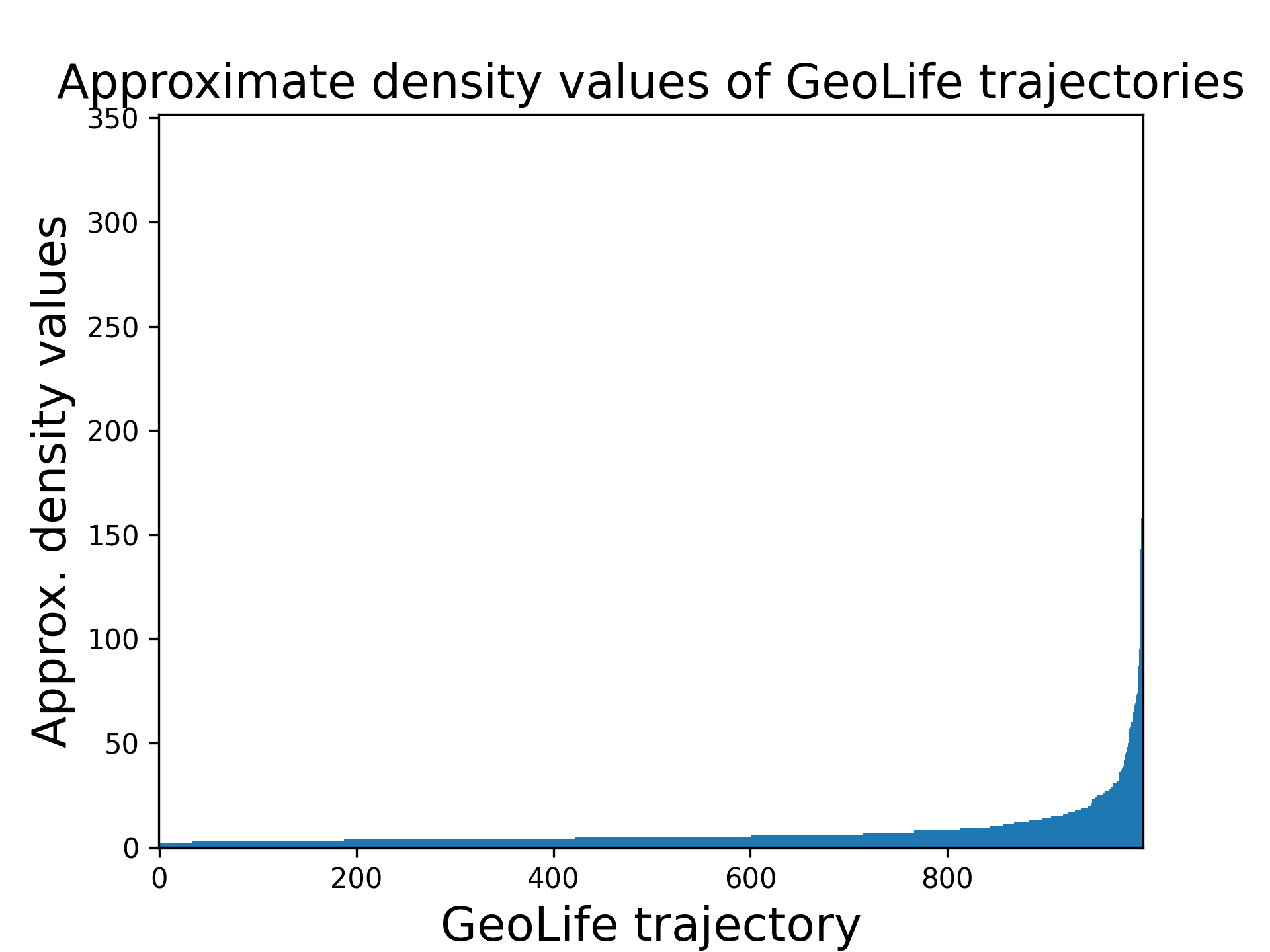} 
    \end{minipage}

\caption{Distributions of approximate density values of trajectories in Vessel-M, Vessel-Y, Truck, Bus, Taxi, and GeoLife data sets}
\label{fig:distribution1}
\end{figure}

\clearpage

\begin{figure}[!htb]

\centering
    \begin{minipage}{0.5\textwidth}
        \centering
        \includegraphics[width=\textwidth]{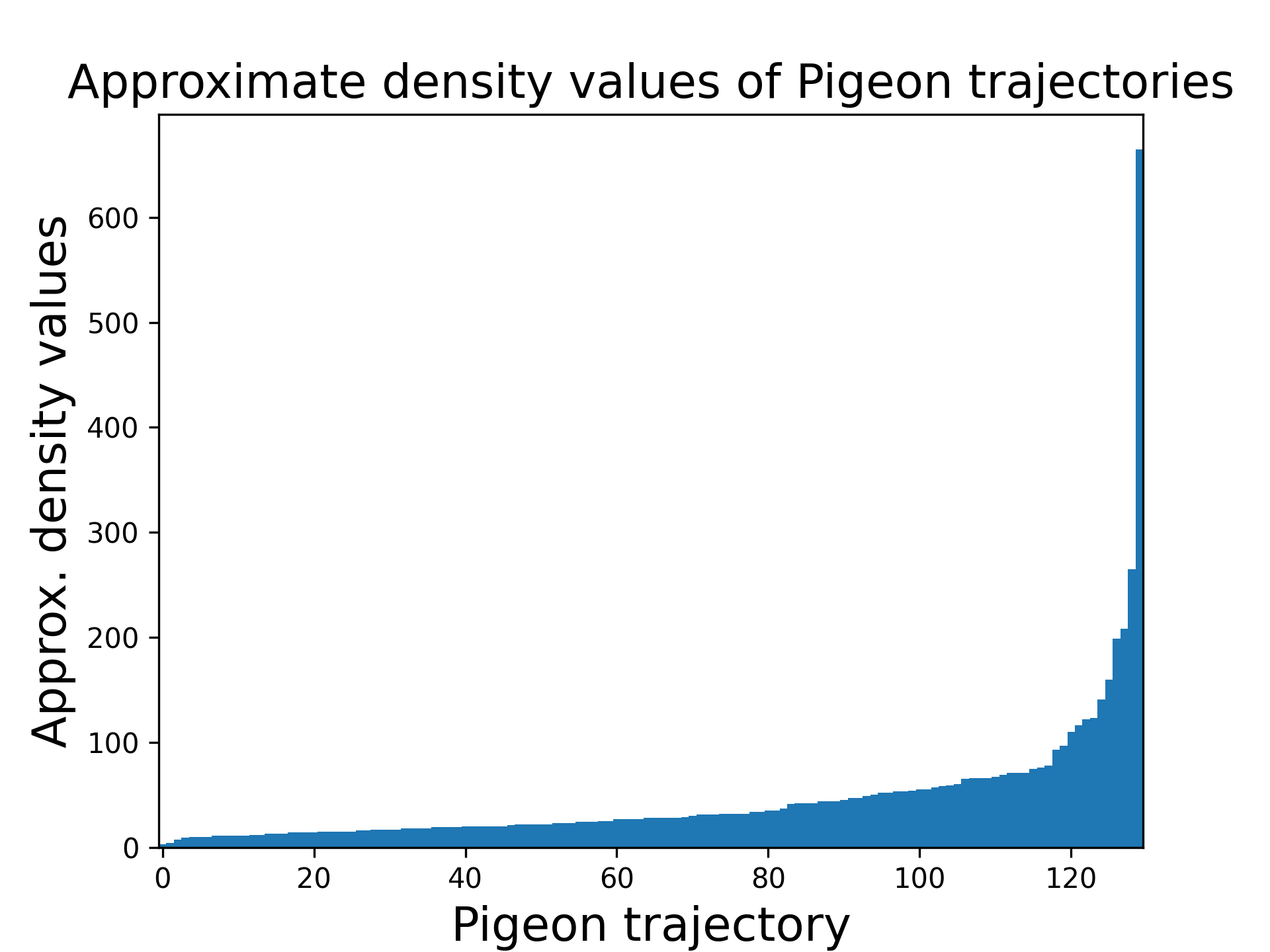} 
    \end{minipage}\hfill
    \begin{minipage}{0.5\textwidth}
        \centering
        \includegraphics[width=\textwidth]{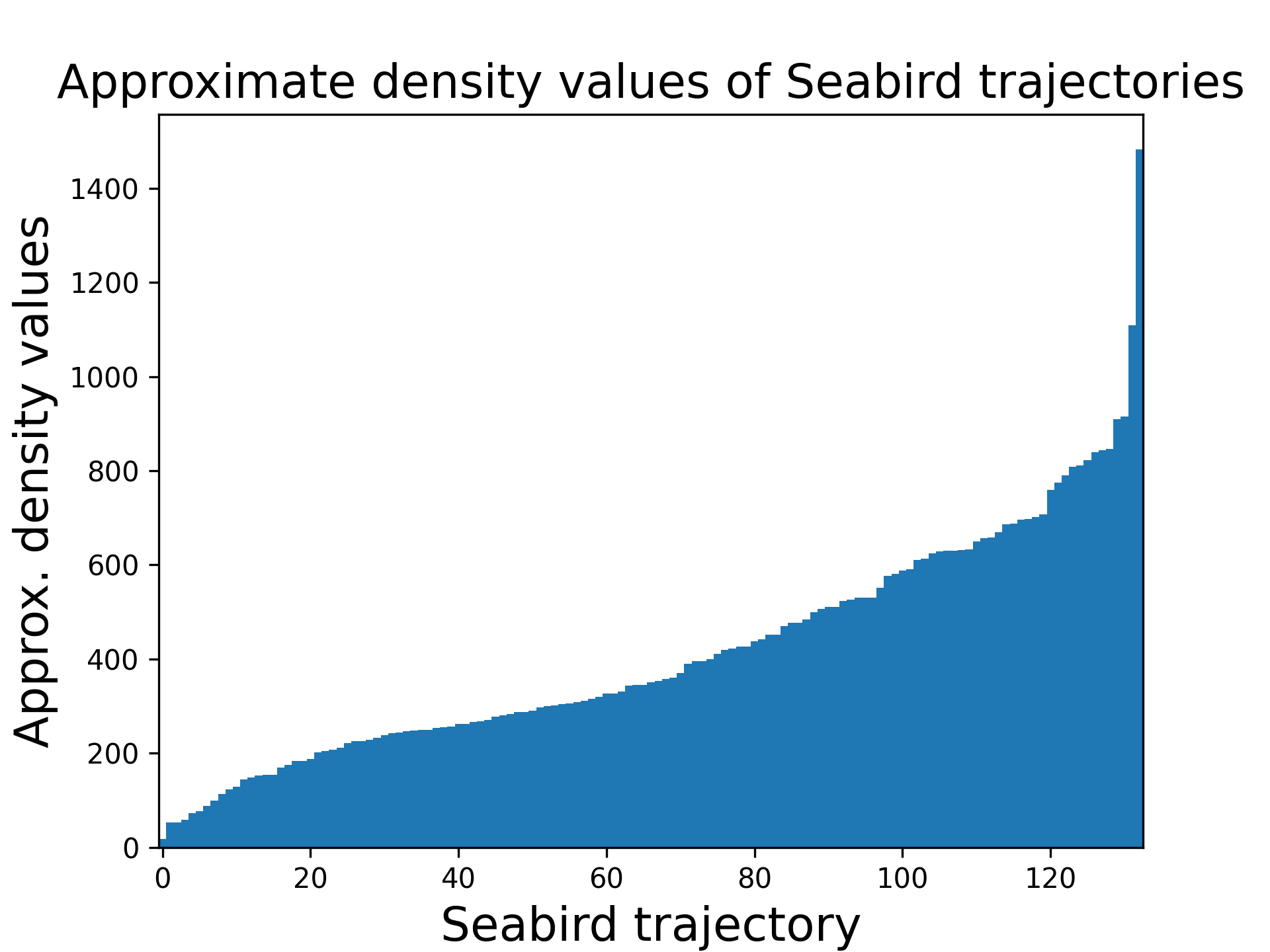} 
    \end{minipage}
    
    \begin{minipage}{0.5\textwidth}
        \centering
        \includegraphics[width=\textwidth]{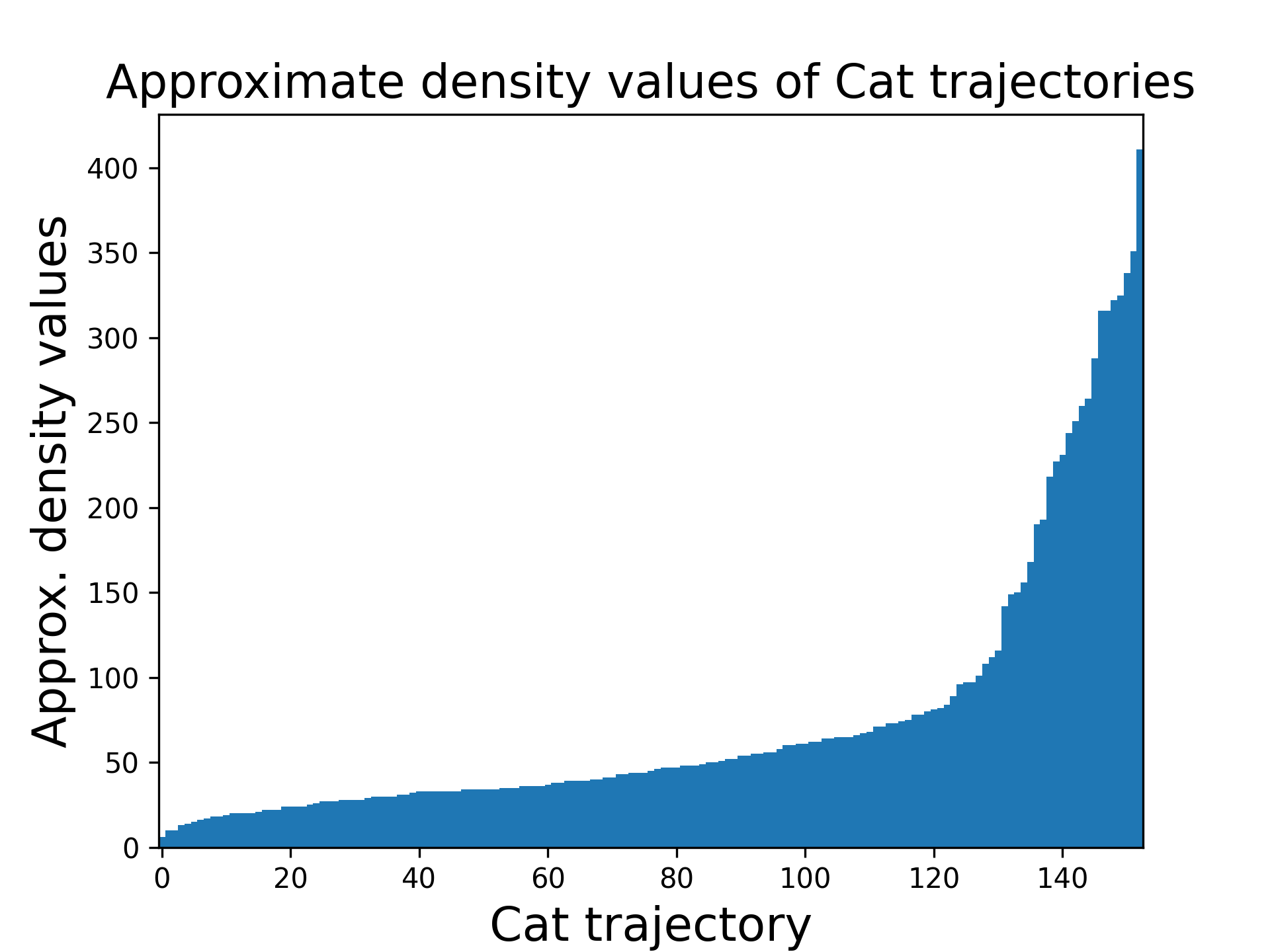} 
    \end{minipage}\hfill
    \begin{minipage}{0.5\textwidth}
        \centering
        \includegraphics[width=\textwidth]{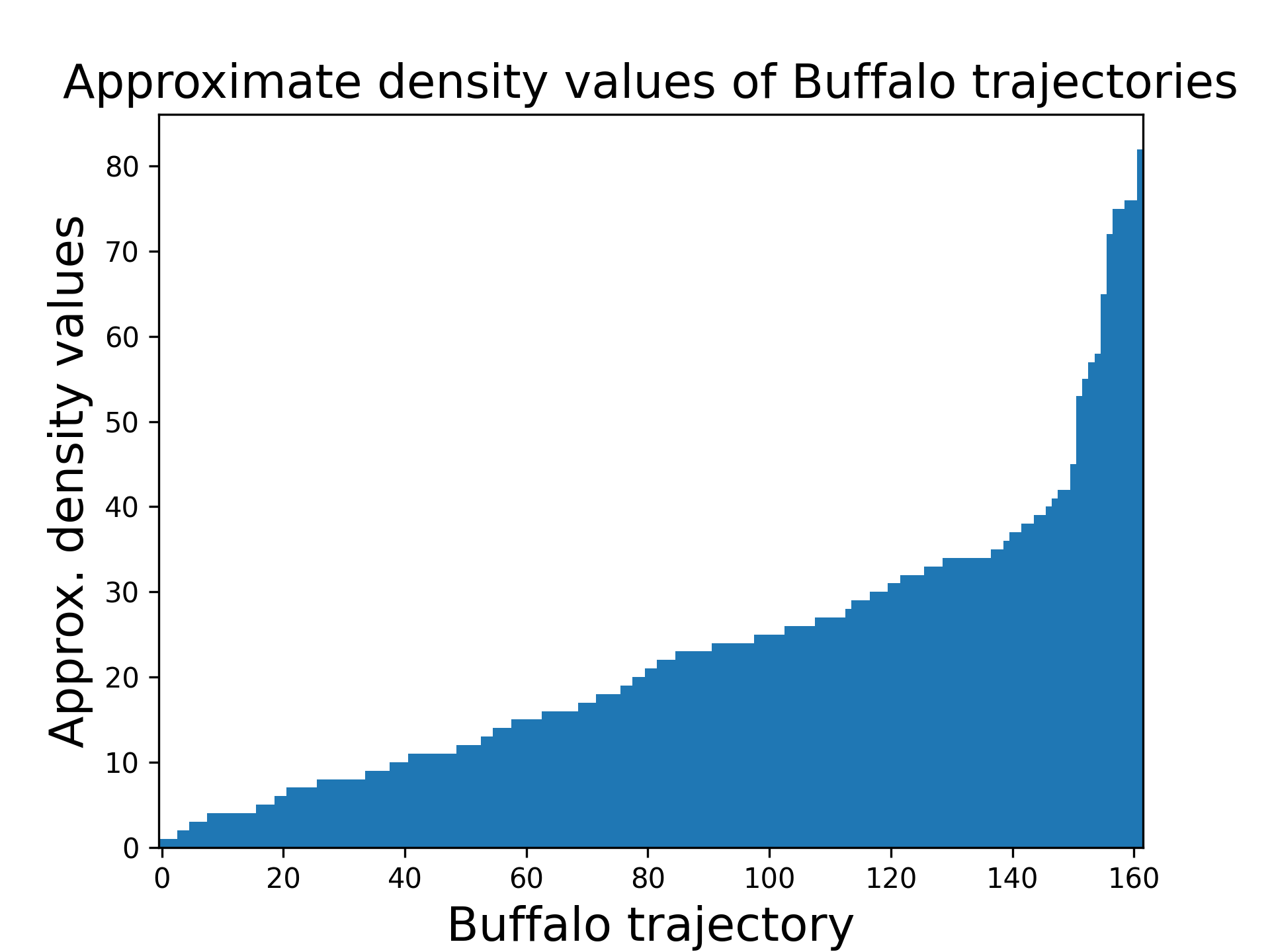} 
    \end{minipage}
    
    \begin{minipage}{0.5\textwidth}
        \centering
        \includegraphics[width=\textwidth]{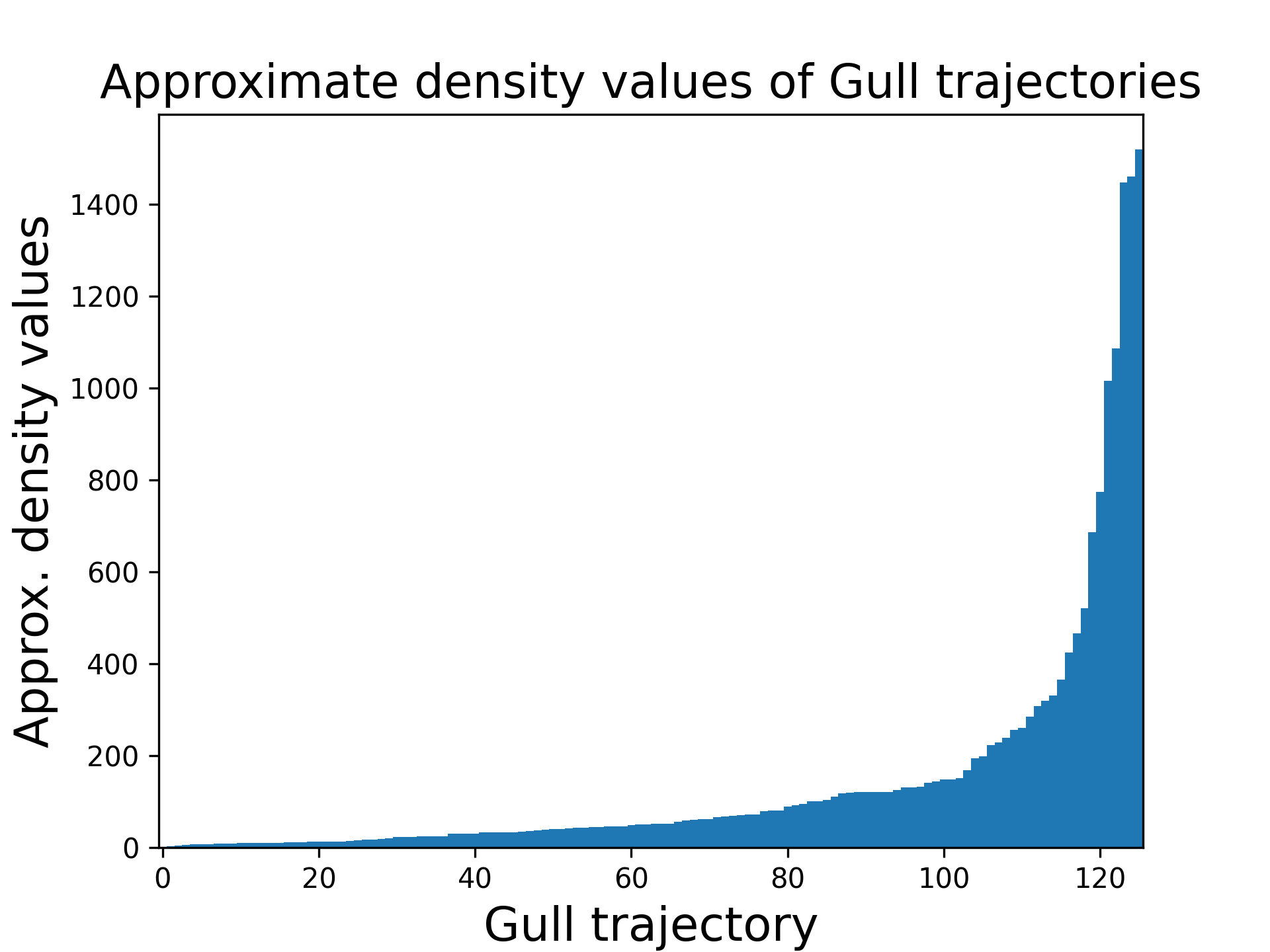} 
    \end{minipage}\hfill
    \begin{minipage}{0.5\textwidth}
        \centering
        \includegraphics[width=\textwidth]{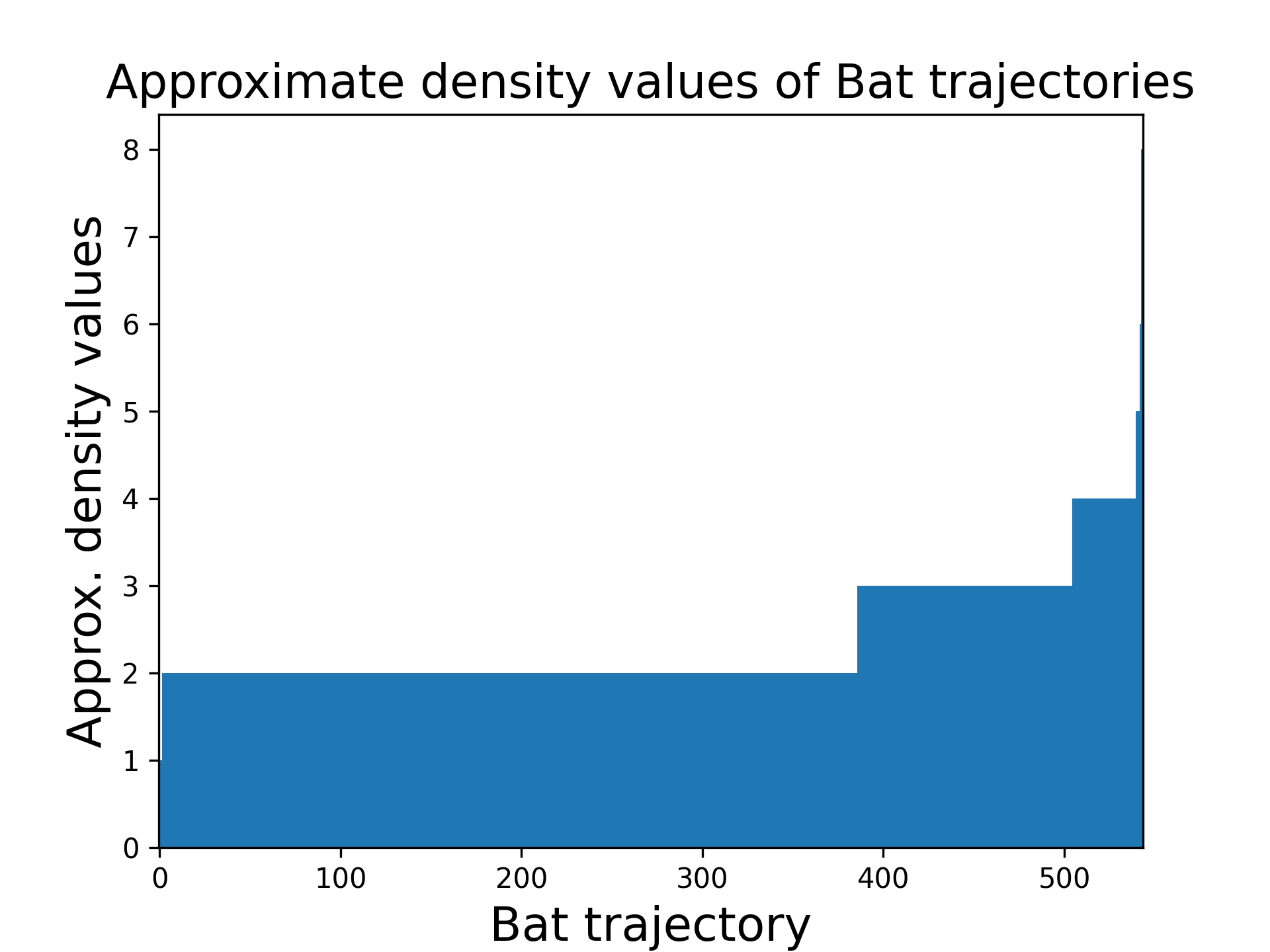} 
    \end{minipage}

\caption{Distributions of approximate density values of trajectories in Pigeon, Seabird, Cat, Buffalo, Gull, and Bat data sets}
\label{fig:distribution2}
\end{figure}

\subsection{Discussion}
In the experiment, we estimated the density values of trajectories in twelve real-world data sets. Our observation is that most of the data sets' estimated density values are low. In eleven data sets (all except Seabird), the median estimate density values are less than $51$. In six data sets, Vessel-Y, Truck, Bus, GeoLife, Pigeon, and Bat, the median $\lambda / n$ ratios are less than $0.04$, where $\lambda$ and $n$ are the estimated density value and the size of the curve. Although the density values are estimates, the notion of low density for trajectories is valuable. The algorithms that perform better when the density value is small, e.g. the algorithm by Driemel et al.~\cite{driemel2012}, can be applied to these data sets to improve efficiency. 


\end{document}